\documentclass[journal,twoside]{IEEEtran}
\usepackage{amsfonts}
\usepackage{graphicx}
\usepackage{color}
\usepackage{amsmath,amsfonts,amssymb,amsthm,epsfig,epstopdf,url,array}
\usepackage{url,textcomp}
\usepackage{authblk}
\usepackage{cite}

\newtheorem{theorem}{Theorem}

\newtheorem{property}{Property}
\begin{document}
\title{Energy-Efficient Optimization for HARQ Schemes over Time-Correlated Fading Channels}
\author{Zheng~Shi,
        Shaodan~Ma,
        Guanghua~Yang,
        and Mohamed-Slim Alouini
\thanks{Copyright (c) 2015 IEEE. Personal use of this material is permitted. However, permission to use this material for any other purposes must be obtained from the IEEE by sending a request to pubs-permissions@ieee.org.}
\thanks{This work was supported in part by National Natural Science Foundation of China under grants 61601524 and 61671488, in part by the Special Fund for Science and Technology Development in Guangdong Province under Grant No. 2016A050503025, in part by the Research Committee of University of Macau under grants MYRG2014-00146-FST and MYRG2016-00146-FST, and in part by the Macau Science and Technology Development Fund under grants 091/2015/A3 and 020/2015/AMJ.}
\thanks{Zheng Shi is with the School of Electrical and Information Engineering and the Institute of Physical Internet, Jinan University (Zhuhai campus), Zhuhai 519070, China (e-mail:shizheng0124@gmail.com).}
\thanks{Shaodan Ma is with the Department of Electrical and Computer Engineering, University of Macau, Macao (e-mail: shaodanma@umac.mo).}
\thanks{Guanghua Yang is with Institute of Physical Internet, Jinan University (Zhuhai Campus), Zhuhai, China (e-mail: ghyang@jnu.edu.cn).}
\thanks{Mohamed-Slim Alouini is with the Computer, Electrical, and Mathematical Science and Engineering (CEMSE) Division, King Abdullah University of Science and Technology (KAUST) Thuwal, Makkah Province, Saudi Arabia (e-mail:slim.alouini@kaust.edu.sa).}
\thanks{The corresponding author is Guanghua Yang.}
}
\markboth{IEEE TRANSACTIONS ON VEHICULAR TECHNOLOGY,~Vol.~xx, No.~x, xxxx~2018}
{SHI \MakeLowercase{\textit{et al.}}: Energy-Efficient Optimization for HARQ Schemes over Time-Correlated Fading Channels}
\maketitle
\begin{abstract}
Energy efficiency of three common hybrid automatic repeat request (HARQ) schemes including Type I HARQ, HARQ with chase combining (HARQ-CC) and HARQ with incremental redundancy (HARQ-IR), is analyzed and joint power allocation and rate selection to maximize the energy efficiency is investigated in this paper. Unlike prior literature, time-correlated fading channels is considered and two widely concerned quality of service (QoS) constraints, i.e., outage and goodput constraints, are also considered in the optimization, which further differentiates this work from prior ones. Using a unified expression of asymptotic outage probabilities, optimal transmission powers and optimal rate are derived in closed-forms to maximize the energy efficiency while satisfying the QoS constraints. These closed-form solutions then enable a thorough analysis of the maximal energy efficiencies of various HARQ schemes. It is revealed that with low outage constraint, the maximal energy efficiency achieved by Type I HARQ is $\frac{1}{4\ln2}$ bits/J, while HARQ-CC and HARQ-IR can achieve the same maximal energy efficiency as $\frac{\kappa_\infty}{4\ln2}$ bits/J where $\kappa_\infty = 1.6617$. Moreover, time correlation in the fading channels has a negative impact on the energy efficiency, while large maximal allowable number of transmissions is favorable for the improvement of energy efficiency. The effectiveness of the energy-efficient optimization is verified by extensive simulations and the results also show that HARQ-CC can achieve the best tradeoff between energy efficiency and spectral efficiency among the three HARQ schemes.

\end{abstract}
\IEEEpeerreviewmaketitle
\hyphenation{HARQ}
\section{Introduction}\label{sec:int}
\IEEEPARstart{T}{he} past few years have witnessed an explosive growth in wireless data traffic and the number of mobile terminals. We are therefore obliged to continuously enhance the spectral/energy efficiency of wireless communication systems to meet the ever-increasing demand. Towards this end, adaptive modulation and coding (AMC) and hybrid automatic repeat request (HARQ) are recognized as two promising techniques so far \cite{Sassioui2016harq,li2015throughput}. In particular, AMC is an efficient physical-layer technique that adapts the modulation type, the transmission rate and even the transmission power to channel state information (CSI) available at the transmitter \cite{harsini2011analysis}. However, with only statistical or limited CSI at the transmitter, HARQ turns out to be more efficient to offer higher throughput and/or lower outage probability with the aid of multiple transmissions adaptive to channel conditions \cite{choi2013onenergyamc}. Specifically, HARQ enables reliable and robust data transmissions through leveraging forward error correction at the physical layer as well as automatic repeat request at the link layer. Generally, HARQ techniques are classified into three types based on the encoding and decoding operations at the transceivers, i.e., Type I HARQ, HARQ with chase combining (HARQ-CC) and HARQ with incremental redundancy (HARQ-IR) \cite{caire2001throughput}. The key difference among them is that Type I HARQ discards the erroneously received sub-codewords and performs memoryless decoding at the receiver, whereas HARQ-CC and HARQ-IR utilize the erroneous sub-codewords for subsequent decoding through chase combining and code combining, respectively. They provide various tradeoffs between performance and complexity and have a wide range of applications in wireless communications.

To further boost the performance of HARQ schemes when statistical/limited CSI is available, optimal design of HARQ system has sparked of considerable research interests lately. The majority of the optimal designs in the literature aim to maximize the spectral efficiency which is commonly quantified by long-term average throughput (LTAT) \cite{seong2011rate,kim2013low,wu2010performance,chelli2013performance,larsson2014throughput,
larsson2016throughput,szczecinski2013rate,choi2013energydelay,chelli2014performance}. For example, cooperative HARQ scheme is investigated and optimal rate and transmission powers are found to maximize the LTAT through numerical search in \cite{seong2011rate}. Noticing no closed-form solution and heavy computational overhead in \cite{seong2011rate}, two suboptimal algorithms are then proposed for rate selection in \cite{kim2013low} to substantially reduce the computational complexity while guaranteeing a comparable performance to the globally optimal solution. By deriving the throughput using Laplace transform, a parameterization-based method is developed to find the optimal rate in (semi-)closed-form to maximize the throughput for various HARQ schemes in \cite{larsson2014throughput}. Moreover, when outdated CSI is available at the transmitter, the optimal transmission rate to maximize the LTAT of HARQ-IR is found by using dynamic programming in \cite{szczecinski2013rate}. Numerical results demonstrate that a notable throughput gain can be achieved by this optimal design.

Apart from spectral efficiency, there is another performance metric which becomes of increasing concern in wireless communications, that is, energy efficiency \cite{li2011energy,wu2014energy,chelli2014performance,choi2013energyfeed,
jabi2016energy,hu2016weighted,basit2015efficient,ge2017energy,zi2016energy,ge2016multi}. It is particularly important in energy limited networks, e.g., internet-of-things (IoT) networks. However, there are very few optimal designs aiming at energy efficiency maximization in the literature. Specifically, in \cite{wu2014energy}, three optimum schemes are proposed to achieve various tradeoffs between the energy efficiency and the spectral efficiency for Type I HARQ under quasi-static fading channels. A fixed power transmission scheme is proposed to maximize the energy efficiency of HARQ-IR under independent fading channels and the optimal power is found numerically in \cite{chelli2014performance}. On the other hand, it is revealed in \cite{choi2013energyfeed,jabi2016energy} that the maximal achievable energy efficiency of Type I HARQ is $\frac{1}{{\rm e}\ln2}$ bits/J, while HARQ-CC and HARQ-IR can achieve the same maximal energy efficiency as $\frac{1}{\ln2}$ bits/J. Unfortunately, these prior analyses and optimal designs rarely consider quality of service (QoS) constraints (e.g., outage probability and throughput constraints) while maximizing the energy efficiency. Such QoS constraints are practical and usually should be satisfied for practical applications. They would definitely deteriorate the energy efficiency due to the shrinking of the feasible region in the optimization. Moreover, all the energy-efficient designs are applicable to either quasi-static fading channels or independent fading channels, but not optimal to time-correlated fading channels which usually occur in low-to-medium mobility environments \cite{kim2011optimal,jin2011optimal}. Time-correlation among fading channels usually causes negative impacts on the system performance, and it has been intensively studied for various systems in the literature. Specifically, the performance of HARQ-IR over time-correlated Rayleigh fading channels is investigated in \cite{shi2015analysis} through polynomial fitting technique. The analytical results are further extended to cooperative HARQ-IR over time-correlated Nakagami-m fading channels in \cite{shi2016inverse}. Most of prior works are conducted for the analysis of either outage probability or spectral efficiency. To our best knowledge, analysis of energy-efficiency and its optimization for HARQ schemes over time-correlated fading channels has not been discussed and solved yet.

Here we thus take a step further to investigate energy-efficient optimization for various HARQ schemes over time-correlated fading channels with QoS constraints, particularly, a tolerable outage constraint and a minimum goodput constraint. A joint power allocation and rate selection scheme is developed to maximize the energy efficiency while satisfying the QoS constraints. With a unified expression of asymptotic outage probabilities, the optimal transmission powers and rates for the three common HARQ schemes including Type I HARQ, HARQ-CC and HARQ-IR are derived in closed-forms. The closed-form optimal solutions then enable a thorough analysis of the maximum energy efficiency when time-correlated fading channels and QoS constraints are considered. It is found that the channel time correlation has a negative impact on the energy efficiency, while large maximal allowable number of transmissions is beneficial to the energy efficiency. More importantly, it is revealed that with low outage constraint, the maximal energy efficiency of Type I HARQ is $\frac{1}{4\ln2}$ bits/J, while both HARQ-CC and HARQ-IR can achieve a maximal energy efficiency of $\frac{\kappa_\infty}{4\ln2}$ bits/J where $\kappa_\infty = 1.6617$.
Noticing that the maximization of the energy efficiency and the spectral efficiency are generally two conflicting goals, the tradeoff between the spectral efficiency and the energy efficiency is also discussed. Numerical results finally demonstrate the effectiveness of our energy-efficient optimization and show that HARQ-CC in fact can achieve a better tradeoff between the energy efficiency and the spectral efficiency than Type I HARQ and HARQ-IR.

The rest of this paper is organized as follows. Section \ref{sec:sysmod} introduces the system model. An energy-efficient optimization is then proposed for the three HARQ schemes and the optimal transmission powers are found in closed-form in Section \ref{sec:endesign}. In Section \ref{sec:opt_rate_ee}, the optimal rates are derived in closed-form for three HARQ schemes and their corresponding optimal energy efficiencies are thoroughly analyzed. Section \ref{sec:numerical} presents the numerical results for verification and investigates the tradeoff between the energy efficiency and the spectral efficiency of our design. Finally, Section \ref{sec:con} concludes this paper.

\section{System Model}\label{sec:sysmod}
This paper considers a point-to-point wireless communication system with HARQ transmissions. Following HARQ protocol, the source first encodes $N_b$ information bits into a mother codeword of length $L {N_s}$, where $L$ denotes the maximal allowable number of transmissions. The codeword is then punctured into $L$ sub-codewords, each with the same length $N_s$. These sub-codewords are sequentially transmitted to the destination through multiple rounds. Notice that in Type I HARQ and HARQ-CC, the same sub-codeword is transmitted in each HARQ round, while in HARQ-IR a different sub-codeword with new parity bits is transmitted in each HARQ round \cite{wu2010performance}. For all the three HARQ schemes, the received signal ${\bf y}_l $ in the $l$th HARQ round at the destination can be unified as
\begin{equation}\label{eqn:recieved_signal_model}
  {\bf y}_l = \sqrt{P_l} h_l {\bf x}_l + {\bf z}_l, \ \ l=1,2,\cdots,L,
\end{equation}
where ${\bf x}_l$ denotes a sub-codeword with length of $N_s$ and each symbol of ${\bf x}_l$ follows Gaussian distribution with unit average energy, i.e., ${\rm E}\{{{\bf x}_l}^{\rm T}{{\bf x}_l}\}=N_s$; $P_l$ stands for the transmission power in the $l$th HARQ round; ${\bf z}_l$ denotes complex Gaussian white noise vector whose elements have zero mean and unit variance; and $h_l$ represents Rayleigh fading channel coefficient in the $l$th HARQ round. Under low-to-medium mobility environments, wireless communications usually experiences time-correlated block-fading channels. To accommodate this channel time correlation, a widely used time-correlated Rayleigh fading channel model given in \cite{kim2011optimal} is adopted here
\begin{equation}\label{eqn:R_k_def}
{h_l} =  {\sigma_l} \left({\sqrt {1 - {\rho ^{2\left( {l  - 1} \right)}}} {\mathfrak h_l} + {\rho ^{l  - 1}}{\mathfrak h _0}}\right), \ \ l=1,2,\cdots,L,
\end{equation}
where $\rho$ and ${\sigma_l}^2$ denote time correlation coefficient and the variance of ${h_l}$, respectively; and ${\mathfrak h _{0}}$, ${\mathfrak h _{1}},\cdots,{\mathfrak h _{L}}$ follow independent circularly-symmetric complex Gaussian distribution with zero mean and unit variance, i.e., ${\mathfrak h_{0}}, {\mathfrak h_{l}} \sim \mathcal{CN}\left( {0,1} \right)$.
This model is general and covers quasi-static fading channels (i.e., fully correlated fading channels where $h_1=h_2=\cdots=h_L$ ) and fast fading channels (i.e., independent fading channels where  $h_1$, $h_2$, $\cdots$, $h_L$ are independent) as special cases with $\rho=1$ and $\rho=0$, respectively. In particular, it is worth highlighting that the fixed-rate HARQ-IR is considered in this paper for the sake of simplicity and fair comparison, that is, the length of each sub-codeword delivered in HARQ-IR scheme keeps fixed throughout HARQ rounds. Hereby, as the length of each sub-codeword is $N_s$ for all the HARQ schemes, the initial transmission rate that denotes the data transmission rate in the first HARQ round is $R=\frac{N_b}{N_s}$.

It is assumed that perfect instantaneous CSI is available at the receiver, but only the statistical CSI is available at the transmitter. To boost the communication performance, the transmission powers $P_1,\cdots,P_L$ and the transmission rate $R$ should be properly designed by using the statistical CSI at the transmitter. In this paper, energy efficiency is concerned and our design objective is to maximize the energy efficiency under QoS constraints by jointly optimizing transmission powers and rate. Since the energy efficiency and QoS constraints are generally defined based on the most fundamental performance metric, i.e., outage probability, it is discussed first in the following. For HARQ schemes, an outage event happens when the destination fails to decode the message after $L$ transmissions. Here three common HARQ schemes are particularly discussed: Type I HARQ, HARQ-CC and HARQ-IR. They are differentiated by their encoding and decoding operations at the transceiver. Specifically, Type I HARQ only takes the received signal in the latest round for decoding, HARQ-CC adopts maximum ratio combining of the received signals in multiple rounds for joint decoding, while HARQ-IR adopts code combining of the received signals in multiple rounds for joint decoding. Therefore the outage probabilities of the three HARQ schemes are different and given respectively as (\ref{outage_definition}) at the top of this page \cite{caire2001throughput}.
\begin{figure*}[!t]
  \centering
\begin{equation}
\label{outage_definition}
{p_{out,L}}=\left\{ {\begin{array}{*{20}{l}}
{\Pr\left({{\log }_2} \left( 1+\max\left({{P_1}{\left| {{h_1}} \right|^2}},{{P_2}{\left| {{h_2}} \right|^2}}, \cdots, {{P_L}{\left| {{h_L}} \right|^2}} \right)\right) \le R \right),}&{{\rm{Type}}\;{\rm{I}}-{\rm HARQ}}\\
{\Pr\left({{\log }_2}\left(1+ \sum\nolimits_{l = 1}^L {{P_l}{\left| {{h_l}} \right|^2}} \right) \le R \right),}&{{\rm{HARQ - CC}}}\\
{\Pr \left(\sum\nolimits_{l = 1}^L {{{\log }_2}\left( {1 + {{P_l}{\left| {{h_l}} \right|^2}}} \right)} \le R \right),}&{{\rm{HARQ - IR}}.}
\end{array}} \right.
\end{equation}
\hrulefill
\end{figure*}
The outage probabilities of the three HARQ schemes under time-correlated fading channels have been derived in \cite{shi2015outage,shi2016optimal,shi2017asymptotic}. However, their expressions are too complex and involve complicated special functions, which provide little insights and hinder the optimal design based on them. Fortunately, by assuming Gaussian codes and typical set decoding \cite{caire2001throughput}, the asymptotic outage probabilities of Type I HARQ and HARQ-CC under time-correlated Rayleigh fading channels have been exactly derived in simple forms in \cite{shi2015outage,shi2016optimal}, respectively, while the asymptotic outage probability of HARQ-IR has been exactly derived in closed-form in \cite{shi2017asymptotic}. It has been shown that the asymptotic outage probabilities match well with the exact outage probabilities under low outage region or high SNR. They can be unified as
\begin{equation}\label{eqn:outage_prob}
{p_{out,L}} \simeq {\phi _L}{\left( {\prod\limits_{k = 1}^L {{P_k}} } \right)^{ - 1}},
\end{equation}
where ${p_{out,0}}=1$ if $R>0$ and ${p_{out,0}}=0$ otherwise, ``$\simeq$'' stands for the asymptotically equal operator, and ${\phi _L}$ changes among various HARQ schemes and is explicitly given by
\begin{equation}\label{eqn:Wl_def}
{\phi _L} = \left\{ {\begin{array}{*{20}{l}}
{{{\varsigma _L}}{{\left( {{2^R} - 1} \right)}^L},}&{{\rm{Type}}\;{\rm{I}}-{\rm HARQ}}\\
{\frac{{{{\varsigma _L}}{{\left( {{2^R} - 1} \right)}^L}}}{{\Gamma \left( {L + 1} \right)}},}&{{\rm{HARQ - CC}}}\\
{{\varsigma _L}}{{{g_L}\left( R \right)},}&{{\rm{HARQ - IR}},}
\end{array}} \right.
\end{equation}
where $\phi _0=1$ if $R>0$ and $\phi _0=0$ otherwise. In (\ref{eqn:Wl_def}), $\Gamma(\cdot)$ denotes Gamma function, ${\varsigma _L}$ quantifies the impact of time-correlated fading channels, i.e., ${\varsigma _L} = \frac{{{{\left( {\ell \left( {L,\rho } \right)} \right)}^{ - 1}}}}{{\prod\nolimits_{k  = 1}^L {{{\sigma_k}^2 }} }}$
in which $\ell \left( {L,\rho} \right)$ particularly quantifies the impact of time correlation with
\begin{equation}\label{eqn:time_corr_impa}
\ell \left( {L,\rho} \right) = \left( {1 + \sum\limits_{k = 1}^L {\frac{{{\rho ^{2(k - 1)}}}}{{1 - {\rho ^{2(k   - 1)}}}}} } \right) \prod\limits_{k = 1}^L {\left( {1 - {\rho ^{2(k   - 1)}}} \right)},
\end{equation}
and $\ell \left( {0,\rho} \right) = 1$. Hence ${\varsigma _0}=1$. Moreover, in (\ref{eqn:Wl_def}), if $R=0$ and $L=0$, $(2^R-1)^L=0$ by convention, and ${{g_L}\left( R \right)}$ is given by
\begin{align}\label{eqn:g_l_def}
{{g_L}\left( R \right)} 
& = \frac{1}{{2\pi {\rm{i}}}}\int\nolimits_{{a} - {\rm{i}}\infty }^{{a} + {\rm{i}}\infty } {\frac{{{2^{Rs}}}}{{s{{\left( {s - 1} \right)}^L}}}ds}\notag\\
&= {\left( { - 1} \right)^L} + {2^R}\sum\limits_{k = 0}^{L - 1} {{{\left( { - 1} \right)}^k}\frac{{{{\left( {R\ln 2} \right)}^{L - k - 1}}}}{{(L-k-1)!}}},
\end{align}
with ${\rm i}=\sqrt{-1}$, $a > 1$ and ${{g_0}\left( 0 \right)} = 0$. The unified asymptotic outage probability in (\ref{eqn:outage_prob}) not only provides clear insights but also enables optimal design of the transmission powers and transmission rate analytically to achieve various objectives. It will lead to closed-form optimal solutions and is adopted here for our optimal design. From the simulation results shown later, the optimal design based on the asymptotic outage probabilities can achieve similar performance to that based on the exact outage probabilities through numerical search, but with much lower complexity.

\section{Energy-Efficient Optimization}\label{sec:endesign}
As defined in \cite{li2011energy,chelli2014performance}, energy efficiency quantifies the amount of successfully delivered information per unit energy. Specifically, for HARQ schemes with maximum $L$ transmissions, based on the \emph{renewal-reward theorem} \cite{zorzi1996use}, the energy efficiency ${\eta}_L$ can be written as the ratio of the average number of correctly received bits $\bar N_b={{N_b}\left( {1 - {p_{out,L}}} \right)}$ to the average amount of energy used $\bar {\mathcal E}=\sum\nolimits_{l = 1}^L {{p_{out,l - 1}}{P_l}{\rm E}\{{{\bf x}_l}^{\rm T}{{\bf x}_l}\}}=N_s\sum\nolimits_{l = 1}^L {{p_{out,l - 1}}{P_l}}$, i.e.,
\begin{equation}\label{eqn:power_effi}
{\eta_L} = \frac{{{{\bar N}_b}}}{{\bar {\mathcal E} }} = \frac{{{N_b}\left( {1 - {p_{out,L}}} \right)}}{{{N_s}\bar P}} 
= \frac{\mathcal T_{g}}{{\bar P}},
\end{equation}
where $\bar P = \sum\nolimits_{l = 1}^L {{p_{out,l - 1}}{P_l}}$ is the average total transmission power and $\mathcal T_{g}=R(1-{p_{out,L}})$. It is worth noting that $\mathcal T_{g}$ is frequently termed as the goodput/effecitve rate, which is an important performance metric to evaluate the throughput of HARQ schemes \cite{rui2008combined,zheng2005optimizing}. As pointed out in \cite{rui2008combined}, the goodput denotes the number of bits successfully delivered per packet transmission. It is asymptotically equivalent to the spectral efficiency under high SNR. Notice that the spectral efficiency denotes the average number of successfully delivered bits per channel use and its definition will be given later.


With statistical CSI at the transmitter, the transmission powers $P_1,\cdots,P_L$ and the transmission rate $R$ could be optimally designed to maximize the energy efficiency ${\eta_{L}}$. In the literature, some optimal/sub-optimal designs have been proposed for HARQ schemes \cite{wu2014energy,chelli2014performance}. Unfortunately, most of them are applicable for independent fading channels and rarely consider QoS constraints. Considering the wide occurrences of time-correlated fading channels and QoS requirements in practice, we take a step forward to incorporate both of them in the energy-efficient optimization. Two widely concerned QoS constraints are particularly considered here. One is the outage constraint, i.e., ${{p_{out,L}} \le \varepsilon }$, and the other is minimum goodput constraint \footnote{The goodput constraint is asymptotically equivalent to the spectral efficiency constraint\cite{rui2008combined}. Since the goodput expression is simpler than the spectral efficiency, the consideration of goodput constraint will simplify the optimization and lead to close-form solution without loss of the practical significance.}, i.e., ${\mathcal T_{g} \ge \mathcal T_0}$. With these QoS constraints, the optimum design of transmission powers and transmission rate to maximize the energy efficiency can be formulated as
\begin{equation}\label{eqn:joint_adaption}
\begin{array}{*{20}{l}}
{\mathop {\rm \max }\limits_{{P_1}, \cdots {P_L},R} }&{{\eta_L} }\\
{{\rm{subject}}\,{\rm{to}}}&{{p_{out,L}} \le \varepsilon }\\
{}&{{\mathcal T_g} \ge \mathcal T_0}\\
{}&{{P_l} \ge 0,\quad 1 \le l \le L}\\
{}&{R \ge 0},
\end{array}
\end{equation}
where $\varepsilon$ and $\mathcal T_0$ denote the maximal allowable outage probability and the minimum required goodput, respectively. It is clear that (\ref{eqn:joint_adaption}) is a fractional programming problem. Due to the complicated expressions of the energy efficiency and outage probability under time-correlated fading channels, the optimal design is very challenging and it is difficult to find the optimal solution directly. But by introducing an auxiliary variable $\alpha={p_{out,L}}$ (\emph{target outage probability}), the fractional programming problem can be reformulated as
\begin{equation}\label{eqn:opt_prob_reform}
\begin{array}{*{20}{l}}
{\mathop {{\rm{\max}}}\limits_{{P_1}, \cdots {P_L},R,\alpha} }&{\frac{R(1-\alpha)}{{\bar P}}}\\
{{\rm{subject}}\,{\rm{to}}}&{{p_{out,L}} = \alpha }\\
{}&{0 \le \alpha  \le \varepsilon }\\
{}&{R(1-\alpha) \ge {\mathcal T}_0}\\
{}&{{P_l} \ge 0,\quad 1 \le l \le L}\\
{}&{R \ge 0},
\end{array}
\end{equation}
which can be further decomposed into three subproblems equivalently with regard to transmission powers $P_1,\cdots,P_L$, target outage probability $\alpha$ and transmission rate $R$, respectively \cite[Eqs.11-13]{palomar2006tutorial}. It should be noted that this equivalent decomposition holds without the necessity of any conditions. They can be solved individually in closed-forms in the following.


\subsection{Optimal Power Allocation}
Given the transmission rate $R$ and the target outage probability $\alpha$, the problem in (\ref{eqn:opt_prob_reform}) reduces to the design of the transmission powers as
\begin{equation}\label{eqn:probem_decom_leve1}
\begin{array}{*{20}{l}}
{\mathop {\rm \min }\limits_{{P_1}, \cdots {P_L}} }&{\bar P}\\
{{\rm{subject}}\,{\rm{to}}}&{{p_{out,L}} = \alpha }\\
{}&{{P_l} \ge 0,1 \le l \le L}.
\end{array}
\end{equation}
This optimization problem is similar to the power allocation problem in \cite{shi2016optimal}, except that the equality outage constraint ${{p_{out,L}} = \alpha }$ is changed as inequality outage constraint ${{p_{out,L}} \le \alpha }$. It has been proved in \cite{shi2016optimal} that the minimal average total transmission power with inequality outage constraint is achieved at the outage boundary, i.e., ${{p_{out,L}} = \alpha }$. Therefore, the closed-form optimal power solution in \cite{shi2016optimal} is applicable to our power design problem in (\ref{eqn:probem_decom_leve1}). Specifically, as shown in \cite{shi2016optimal}, the optimal powers are given as functions of the transmission rate $R$ and target outage probability $\alpha$ as
\begin{align}\label{eqn:P_L_star_sec}
{P_L^*}&{ = {{\left( {\frac{{{\phi _L}\prod\limits_{k = 2}^L {{{\left( {\frac{{2{\phi _{k - 1}}}}{{{\phi _{k - 2}}}}} \right)}^{{2^{1 - k}}}}} }}{{{2^{L - 1}}\alpha{\phi _{L - 1}} }}} \right)}^{\frac{{{2^{L - 1}}}}{{{2^L} - 1}}}}},
\end{align}
\begin{equation}\label{eqn:kkt1_subs_fin}
P_l^* = \prod\limits_{k = l + 1}^L {{{\left( {\frac{{2{\phi _{k - 1}}}}{{{\phi _{k - 2}}}}} \right)}^{{2^{l - k}}}}} {P_L^*}^{{2^{l - L}}},\quad 1 \le l \le L - 1,
\end{equation}
and the minimal average total transmission power $\bar P^*$ can be obtained in a simple form as
\begin{equation}\label{EQN:AVG_P}
{{\bar P}^*} = \frac{{\left( {{2^L} - 1} \right){\alpha ^{ - \frac{1}{{{2^L} - 1}}}}}}{{{2^{\frac{L}{{1 - {2^{ - L}}}} - 2}}}}{\left( {\prod\limits_{k = 1}^L {{{\left( {\frac{{{\phi _k}}}{{{\phi _{k - 1}}}}} \right)}^{{2^{ - k}}}}} } \right)^{\frac{1}{{1 - {2^{ - L}}}}}}.
\end{equation}

Clearly from (\ref{eqn:P_L_star_sec}), (\ref{eqn:kkt1_subs_fin}) and (\ref{EQN:AVG_P}), the decrease of the target outage probability $\alpha$ will lead to the increase of transmission powers $P^*_l$ and then the increase of the minimal average total transmission power $\bar P^*$. In addition, it can be found that the minimal average total transmission power $\bar P^*$ becomes irrelevant to the target outage probability $\alpha$ when $L \to \infty$.

\subsection{Optimal Outage Probability}
Now putting the optimal powers (\ref{eqn:P_L_star_sec}) and (\ref{eqn:kkt1_subs_fin}) into the original optimization problem (\ref{eqn:opt_prob_reform}), the original problem can be reduced to the optimization of two variables of the transmission rate $R$ and target outage probability $\alpha$ as
\begin{equation}\label{eqn:probe_leve2}
\begin{array}{*{20}{l}}
{\mathop {\rm \max}\limits_{R,\alpha} }&\frac{R(1-\alpha)}{{{\bar P}^*}}\\
{{\rm{subject}}\,{\rm{to}}}&{0 \le \alpha \le \varepsilon }\\
{}&{R(1-\alpha) \ge \mathcal T_0}\\
{}&{R \ge 0}.
\end{array}
\end{equation}
With (\ref{EQN:AVG_P}), the objective function of (\ref{eqn:probe_leve2}) can be rewritten as
\begin{equation}\label{eqn:rewri_obje}
 \frac{{R(1 - \alpha )}}{{{{\bar P}^*}}} = \underbrace {\frac{{{2^{\frac{L}{{1 - {2^{ - L}}}} - 2}}}}{{\left( {{2^L} - 1} \right)}}}_{\triangleq \psi} \underbrace {(1 - \alpha ){\alpha ^{\frac{1}{{{2^L} - 1}}}}}_{\triangleq f\left( \alpha  \right)} {\frac{R}{{{{\left( {\prod\limits_{k = 1}^L {{{\left( {\frac{{{\phi _k}}}{{{\phi _{k - 1}}}}} \right)}^{{2^{ - k}}}}} } \right)}^{\frac{1}{{1 - {2^{ - L}}}}}}}}}.
\end{equation}
It is clear that the target outage probability is only involved in the term $f\left( \alpha  \right)$ in (\ref{eqn:rewri_obje}). When given the transmission rate $R$, the problem in (\ref{eqn:probe_leve2}) is equivalent to the following target outage probability optimization as
\begin{equation}\label{eqn:tar_out_prob}
\begin{array}{*{20}{l}}
{\mathop {{\rm{\max}}}\limits_{\alpha } }&{f\left( \alpha  \right)}\\
{{\rm{subject\,to}}}&{0 \le \alpha  \le \varepsilon }\\
{}&{R(1 - \alpha ) \ge {\mathcal T_0}},
\end{array}
\end{equation}
whose feasibility and optimal solution can be found in the following theorem.
\begin{theorem}\label{the:optimal_outage}
The optimization problem (\ref{eqn:tar_out_prob}) is infeasible when $R \le \mathcal T_0$. When $R > \mathcal T_0$, the optimal target outage probability is ${\alpha ^*} = {\min \left\{ {\varepsilon ,1 - \frac{{{{\mathcal T_0}}}}{R}},2^{-L} \right\}}$ and the optimal $f\left( {{\alpha ^*}} \right)$ can be written as
\begin{align}\label{eqn:f_alpha_star}
f\left( {{\alpha ^*}} \right)
&=  {{\frac{{{{\mathcal T_0}}}}{R}{{\left( {1 - \frac{{{{\mathcal T_0}}}}{R}} \right)}^c}}}\left( {{\chi}\left( {R - {{\mathcal T_0}}} \right) - {\chi}\left( {R - \frac{{{{\mathcal T_0}}}}{{1 - \Delta }}} \right)} \right) \notag\\
&\quad + {{(1 - \Delta ){\Delta ^c}}}{\chi }\left( {R - \frac{{{{\mathcal T_0}}}}{{1 - \Delta }}} \right),
\end{align}
where $c = \frac{1}{{{2^L} - 1}}$, $\Delta  = \min \left\{ {\varepsilon ,{2^{ - L}}} \right\}$, and $\chi (t)$ is an indicator function as
\begin{equation}\label{eqn:indi_fundef}
\chi (t){\rm{ = }}\left\{ {\begin{array}{*{20}{c}}
{\rm{0}}&{t < 0}\\
{\rm{1}}&{t \ge 0}
\end{array}} \right..
\end{equation}

\end{theorem}
\begin{proof}
Please see Appendix \ref{app:proof_tarout}.
\end{proof}

\subsection{Optimal Rate Selection}
After determining the optimal target outage probability $\alpha^*$ and considering the feasible region of $R > \mathcal T_0$ in Theorem \ref{the:optimal_outage}, the energy-efficient optimization in (\ref{eqn:probe_leve2}) is finally reduced to the optimal rate selection as
\begin{equation}\label{eqn:rate_sele_prob}
\begin{array}{*{20}{l}}
{\mathop {{\rm{\max}}}\limits_R }&{{\eta_L } = \psi f\left( {{\alpha ^*}} \right){\frac{R}{{{{\left( {\prod\limits_{k = 1}^L {{{\left( {\frac{{{\phi _k}}}{{{\phi _{k - 1}}}}} \right)}^{{2^{ - k}}}}} } \right)}^{\frac{1}{{1 - {2^{ - L}}}}}}}}}}\\
{{\rm{subject\,to}}}&{R > {{\mathcal T_0}}.}
\end{array}
\end{equation}

Plugging (\ref{eqn:f_alpha_star}) into (\ref{eqn:rate_sele_prob}), although the optimal transmission rate $R^*$ can be computed numerically through one dimensional search, it is lacking of clear insights. In this paper, we aim to derive the optimal rate in closed-form and extract clear insights for energy-efficient optimization. Since $\phi _k$ in the objective function depends on the transmission rate $R$ as shown in (\ref{eqn:Wl_def}) and changes among different HARQ schemes, the optimal rate selection for various HARQ schemes should be investigated individually and will be discussed one by one in the next section.


\section{Optimal Rate and Optimal Energy Efficiency}\label{sec:opt_rate_ee}
\subsection{Type I HARQ}\label{sec:harq_I}
\subsubsection{Optimal Rate}
Putting (\ref{eqn:Wl_def}) into the objective function in (\ref{eqn:rate_sele_prob}), the energy efficiency of Type I HARQ can be written as
\begin{equation}\label{eqn:power_eff_rew_op}
{\eta _{I,L}} 
= \psi {\theta _{L}} \frac{f\left( {{\alpha ^*}} \right)R}{{{2^R} - 1}},
\end{equation}
where
\begin{align}\label{eqn:theta_I}
{\theta _L} &= {\left( {\prod\limits_{k = 1}^L {{{\left( {\frac{{{\varsigma _{k - 1}}}}{{{\varsigma _k}}}} \right)}^{{2^{ - k}}}}} } \right)^{\frac{1}{{1 - {2^{ - L}}}}}}\notag\\
&= {\left( {\prod\limits_{k = 1}^L {{{\left( {\frac{{\ell \left( {k,\rho } \right){{\sigma_k}^2 }}}{{\ell \left( {k - 1,\rho } \right)}}} \right)}^{{2^{ - k}}}}} } \right)^{\frac{1}{{1 - {2^{ - L}}}}}}.
\end{align}
Substituting (\ref{eqn:f_alpha_star}) into (\ref{eqn:power_eff_rew_op}) yields
\begin{align}\label{eqn:power_type_I_rew1}
{\eta _{I,L}}&= {{\psi {\theta _L}}}{{{{\mathcal T_0}}}}\frac{{{{\left( {1 - \frac{{{{\mathcal T_0}}}}{R}} \right)}^c}}}{{{2^R} - 1}}\left( {{\chi}\left( {R - {{\mathcal T_0}}} \right) - {\chi}\left( {R - \frac{{{{\mathcal T_0}}}}{{1 - \Delta }}} \right)} \right) \notag\\
&\quad + {{\psi {\theta _L}}}{{(1 - \Delta ){\Delta ^c}}}\frac{R}{{{2^R} - 1}}{\chi}\left( {R - \frac{{{{\mathcal T_0}}}}{{1 - \Delta }}} \right).
\end{align}
From (\ref{eqn:power_type_I_rew1}), we can see that when $R > \frac{\mathcal T_0}{1-\Delta}$, the first term is zero and the energy efficiency reduces to ${\eta _{I,L}}={{\psi {\theta _L}}}{{(1 - \Delta ){\Delta ^c}}}\frac{R}{{{2^R} - 1}}$ which is a decreasing function of $R$. Moreover, the energy efficiency ${\eta _{I,L}}$ is continuous in the domain of $R > \mathcal T_0$. Therefore, considering the constraint of $R > \mathcal T_0$ in (\ref{eqn:rate_sele_prob}), we can conclude that the optimal rate ${R^*}$ to achieve the maximal energy efficiency should lie within the range of $\left({\mathcal T_0}, \frac{\mathcal T_0}{1-\Delta}\right]$. In this range, the second term in (\ref{eqn:power_type_I_rew1}) is zero and the energy efficiency can be simplified as
%
%
%
\begin{align}\label{eqn:power_type_I_rew1_corr}
{\eta _{I,L}}&= {{\psi {\theta _L}}}{{{{\mathcal T_0}}}}\frac{{{{\left( {1 - \frac{{{{\mathcal T_0}}}}{R}} \right)}^c}}}{{{2^R} - 1}}.
\end{align}
Then the problem of rate selection in (\ref{eqn:rate_sele_prob}) is equivalent to a minimization problem as
\begin{equation}\label{eqn:type_i_rate_ref1}
\begin{array}{*{20}{l}}
{\mathop {{\rm{\min}}}\limits_R }&{\Phi \left( R \right) = \left( {{2^R} - 1} \right){\left( {1 - \frac{{{{\mathcal T_0}}}}{R}} \right)^{ - c}}}\\
{{\rm{subject}}{\mkern 1mu}\,{\rm{to}}}&{{{\mathcal T_0}} < R \le \frac{{{{\mathcal T_0}}}}{{1 - \Delta }}}.
\end{array}
\end{equation}
The optimal solution to (\ref{eqn:type_i_rate_ref1}) can be found in closed-form and is shown in the following theorem.
\begin{theorem}\label{the:optimal_rate}
The optimal rate to maximize the energy efficiency while guaranteeing the outage and goodput performance for Type I HARQ is
\begin{equation}\label{eqn:opt_rate_type_i_re}
{R^*} = 
\min\left\{{\frac{{{{\cal T}_0}}}{{1 - \Delta }}}, {{\varphi ^{ - 1}}(0)}\right\},
\end{equation}
where $\varphi^{-1} $ denotes the inverse function with respect to $\varphi \left( R \right) = \ln \left( 2 \right)R\left( {R - {{\mathcal T_0}}} \right){2^R} - c{{\mathcal T_0}}\left( {{2^R} - 1} \right)$ and $\varphi^{-1}(0)$ refers to the zero point of $\varphi(R)$. When $\varepsilon \ge 2^{-L}$, the optimal transmission rate reduces to $R^*=\varphi^{-1}(0)$.
\end{theorem}
\begin{proof}
Please see Appendix \ref{app:proof_convex}.
\end{proof}
It is noteworthy that the zero point $\varphi^{-1}(0)$ can be easily computed since $\varphi \left( R \right) $ is an increasing function of $R$ within the range of $(\mathcal T_0, \infty)$.

\subsubsection{Optimal Energy Efficiency}\label{sec:opt_ee_I}
Putting the optimal rate (\ref{eqn:opt_rate_type_i_re}) into (\ref{eqn:power_type_I_rew1_corr}), the optimal energy efficiency of Type I HARQ can be obtained as
\begin{equation}\label{eqn:power_eff_opt}
{{\eta}  _{I,L}^*} = {{\psi {\theta _L}}}{{{{\mathcal T_0}}}}\frac{{{{\left( {1 - \frac{{{{\mathcal T_0}}}}{{{R^*}}}} \right)}^c}}}{{{2^{{R^*}}} - 1}}  = {{\psi {\theta _L}}}{{{{\mathcal T_0}}}} \frac{{{{\left( {{R^*} - {{\mathcal T_0}}} \right)}^c}}}{{{R^*}^c\left( {{2^{{R^*}}} - 1} \right)}}.
\end{equation}

As aforementioned, the maximal energy efficiency of HARQ schemes over independent Rayleigh fading channels without QoS constraints has been studied in \cite{jabi2016energy}. It has been proved that the maximal energy efficiency of Type I HARQ operating over \emph{independent Rayleigh fading channels without QoS constraints} is $\frac{1}{{\rm e} \ln(2)}$. From (\ref{eqn:power_eff_opt}), we can see that channel time correlation will affect the optimal energy efficiency through the term $\theta_L$. As proved in Appendix \ref{app:time_corr}, $\theta_L$ is a decreasing function of time correlation coefficient $\rho$. In other words, channel time correlation has a negative impact on the optimal energy efficiency and time-correlated fading channels provide lower energy efficiency than time independent fading channels with $\rho=0$. We thus can expect that the maximal energy efficiency of Type I HARQ over time-correlated fading channels with QoS constraints will be lower than $\frac{1}{{\rm e} \ln(2)}$ and will be investigated here.

To proceed with the investigation, we first analyze the monotonic property of the optimal energy efficiency with respect to the maximal number of transmissions $L$. From the original energy efficiency maximization problem in (\ref{eqn:joint_adaption}), we can find the following property.

\begin{property}\label{the:type_I_eff}
The optimal energy efficiencies of all the three HARQ schemes are non-decreasing functions of the maximal number of transmissions $L$ and ${\eta _{L}^*} \le \mathop {\lim }\limits_{L \to \infty}{\eta _{L}^*}  \triangleq  {\eta _{\infty}^*} $. 
%
\end{property}
\begin{proof}
Please see Appendix \ref{app:there2_proof}.
\end{proof}

Notice that this property is applicable to all the three HARQ schemes. For Type I HARQ scheme, the maximal energy efficiency is thus achieved when $L \to \infty$ and is denoted as ${\eta _{I,\infty}^*}$ which can be found as shown in the following theorem.
\begin{theorem}\label{the:ee_I}
Under time-correlated Rayleigh fading channels, the optimal energy efficiency of Type I HARQ with outage and goodput constraints is upper bounded by
\begin{equation}\label{eqn:upper_bound}
{\eta _{I,L}^*} \le {\eta _{I,\infty}^*} = \frac{{{\theta _\infty }}{{\mathcal T_0}}}{{4 \left( {{2^{{{\mathcal T_0}}}} - 1} \right)}},
\end{equation}
where $\theta_\infty \triangleq {\mathop {\lim }\limits_{L \to \infty } {\theta _L}}$ exists. The maximal energy efficiency ${\eta _{I,\infty}^*}$ is a decreasing function of the goodput threshold $\mathcal T_0$. Particularly, for Rayleigh fading channels with unit channel gains as ${\sigma_1}^2=\cdots={\sigma_L}^2=1$, $\theta_\infty \le 1$ and the maximal energy efficiency of Type I HARQ with outage and goodput constraints satisfies
\begin{equation}\label{type_I_maximal}
{\eta _{I,\infty}^*} = \frac{{{\theta _\infty }}{{\mathcal T_0}}}{{4 \left( {{2^{{{\mathcal T_0}}}} - 1} \right)}} \le \mathop {\lim }\limits_{{\mathcal T_0} \to 0} \frac{{{{\mathcal T_0}} }}{{4\left( {{2^{{{\mathcal T_0}}}} - 1} \right)}}=\frac{1}{4\ln(2)} \triangleq {\eta _{I,\infty}^{\max}}.
\end{equation}
\end{theorem}
\begin{proof}
Please see Appendix \ref{app:ee_I}.
\end{proof}

It means that when QoS constraints are considered, the maximal energy efficiency which can be achieved by Type I HARQ is $\frac{1}{4\ln(2)}$.

\subsection{HARQ-CC}\label{sec:harq_cc}
\subsubsection{Optimal Rate}
Similarly to Section \ref{sec:harq_I}, with the optimal powers $P_1^*,\cdots,P_L^*$, the optimal target outage probability $\alpha^*$ and the definition in (\ref{eqn:Wl_def}), the energy efficiency of HARQ-CC in (\ref{eqn:rate_sele_prob}) can be rewritten as
\begin{equation}\label{eqn:power_eff_cc}
{\eta _{CC,L}} 
= {\kappa _L}^{\frac{1}{{1 - {2^{ - L}}}}} \psi {\theta _L}
\frac{f\left( {{\alpha ^*}} \right)R}{{{2^R} - 1}},
\end{equation}
where ${\kappa _L} = \prod\nolimits_{k = 1}^L {{{ {{k}} }^{{2^{ - k}}}}}$.
Noticing that the only difference between (\ref{eqn:power_eff_rew_op}) and (\ref{eqn:power_eff_cc}) is at an additional term ${\kappa _L}^{\frac{1}{{1 - {2^{ - L}}}}}$ involved in (\ref{eqn:power_eff_cc}) and this term is irrelevant to the transmission rate, the optimal transmission rate $R^*$ of HARQ-CC can thus be derived as the same as that for Type I HARQ shown in (\ref{eqn:opt_rate_type_i_re}) in Theorem \ref{the:optimal_rate}.

\subsubsection{Optimal Energy Efficiency}
Accordingly, the optimal energy efficiency of HARQ-CC can be written as
\begin{equation}\label{eqn:power_eff_opt_cc}
{\eta _{CC,L}^{*}} 
={\kappa _L}^{\frac{1}{{1 - {2^{ - L}}}}} {\eta _{I,L}^{*}}.
\end{equation}
Since $\kappa_L > 1$, it thus reveals that HARQ-CC surpasses Type I HARQ in terms of the optimal energy efficiency, i.e., ${\eta _{CC,L}^{*}} > {\eta _{I,L}^{*}}$.

Moreover, based on Property \ref{the:type_I_eff} and Theorem \ref{the:ee_I}, it is easy to get the following result of maximal energy efficiency of HARQ-CC.


\begin{theorem}\label{the:ee_CC}
Under time-correlated Rayleigh fading channels, the optimal energy efficiency of HARQ-CC with outage and goodput constraints is upper bounded by
\begin{equation}\label{eqn:upper_bound_cc}
 {\eta _{CC,L}^*} \le {\eta _{CC,\infty}^*} = \frac{{{\kappa_\infty\theta _\infty }}{{\mathcal T_0}}}{{4 \left( {{2^{{{\mathcal T_0}}}} - 1} \right)}},
\end{equation}
where ${\kappa _\infty }= \mathop {\lim }\limits_{{L} \to \infty} \kappa_L\approx 1.6617$ as proved in Appendix \ref{app:proof_rec}. Particularly, for Rayleigh fading channels with unit channel gains as ${\sigma_1}^2=\cdots={\sigma_L}^2=1$, $\theta_\infty \le 1$ and the maximal energy efficiency of HARQ-CC with outage and goodput constraints satisfies
\begin{equation}
{\eta _{CC,\infty}^*} = \frac{{{\kappa_\infty\theta _\infty }}{{\mathcal T_0}}}{{4 \left( {{2^{{{\mathcal T_0}}}} - 1} \right)}} \le \mathop {\lim }\limits_{{\mathcal T_0} \to 0} \frac{{{{\kappa_\infty\mathcal T_0}} }}{{4\left( {{2^{{{\mathcal T_0}}}} - 1} \right)}} = \frac{\kappa_\infty}{4\ln(2)} \triangleq {\eta _{CC,\infty}^{\max}}.
\end{equation}
\end{theorem}

In other words, when QoS constraints are considered, the maximal energy efficiency of HARQ-CC is $\frac{\kappa_\infty}{4\ln(2)}$ which is higher than that of Type I HARQ.

\subsection{HARQ-IR}\label{sec:harq_ir}
\subsubsection{Optimal Rate}
Putting (\ref{eqn:Wl_def}) into (\ref{eqn:rate_sele_prob}) yields the energy efficiency of HARQ-IR as
\begin{align}\label{eqn:power_eff_IR}
{\eta _{IR,L}} 
&= \psi {\theta _L} \frac{f\left( {{\alpha ^*}} \right)R} {{{{\left( {\prod\limits_{k = 1}^L {{{\left( {\frac{{{g_k}\left( R \right)}}{{{g_{k - 1}}\left( R \right)}}} \right)}^{{2^{ - k}}}}} } \right)}^{\frac{1}{{1 - {2^{ - L}}}}}}}}.
\end{align}
Substituting (\ref{eqn:f_alpha_star}) into (\ref{eqn:power_eff_IR}), it becomes
\begin{multline}\label{eqn:power_eff_ir_rew}
{\eta _{IR,L}} =
 \frac{{{\psi {\theta _L}}}{{{{\mathcal T_0}}}}}{\Lambda \left( R \right)}\left( {{\chi}\left( {R - {{\mathcal T_0}}} \right) - {\chi}\left( {R - \frac{{{{\mathcal T_0}}}}{{1 - \Delta }}} \right)} \right) \\
  + \frac{{{\psi {\theta _L}}}{{(1 - \Delta ){\Delta ^c}}}}{\left( {\prod\limits_{k = 1}^L {{{\left( {\cal G}_k\left( R \right) \right)}^{{2^{ - k}}}}} } \right)^{\frac{1}{{1 - {2^{ - L}}}}}}{\chi}\left( {R - \frac{{{{\mathcal T_0}}}}{{1 - \Delta }}} \right),
\end{multline}
where ${\Lambda \left( R \right) \triangleq {{\left( {1 - \frac{{{{\mathcal T_0}}}}{R}} \right)}^{ - c}}{{\left( {\prod\limits_{k = 1}^L {{{\left( {\frac{{{g_k}\left( R \right)}}{{{g_{k - 1}}\left( R \right)}}} \right)}^{{2^{ - k}}}}} } \right)}^{\frac{1}{{1 - {2^{ - L}}}}}}}$ and ${\cal G}_{k}\left( R \right) \triangleq \frac{{{g_k}\left( R \right)}}{{R{g_{k - 1}}\left( R \right)}}$. When $R > \frac{\mathcal T_0}{1-\Delta}$, the first term in (\ref{eqn:power_eff_ir_rew}) is zero and the energy efficiency reduces as ${\eta _{IR,L}} =
 \frac{{{\psi {\theta _L}}}{{(1 - \Delta ){\Delta ^c}}}}{\left( {\prod\limits_{k = 1}^L {{{\left( {\cal G}_k\left( R \right) \right)}^{{2^{ - k}}}}} } \right)^{\frac{1}{{1 - {2^{ - L}}}}}}$.
As proved in Appendix \ref{app:proof_increasing_g_ratio} that ${\cal G}_{k}\left( R \right) $ is a monotonically increasing function of $R$. Thus when $R > \frac{\mathcal T_0}{1-\Delta}$,  ${\eta _{IR,L}}$ is a decreasing function of $R$. Together with the continuity of ${\eta _{IR,L}}$ at the point $R={\frac{{{{\mathcal T_0}}}}{{1 - \Delta }}}$ and the constraint of $R > \mathcal T_0$ in (\ref{eqn:rate_sele_prob}), we can conclude that the optimal rate to achieve the maximal energy efficiency of HARQ-IR is located within the range of $\left({\mathcal T_0}, \frac{\mathcal T_0}{1-\Delta}\right]$.
%
Accordingly, when ${\mathcal T_0} < R \le \frac{\mathcal T_0}{1-\Delta}$, the energy efficiency ${\eta _{IR,L}}$ can be simplified as
\begin{equation}\label{eqn:eta_ir_l_simp}
{\eta _{IR,L}} =
 \frac{{{\psi {\theta _L}}}{{{{\mathcal T_0}}}}}{\Lambda \left( R \right)}.
\end{equation}
Then the rate selection problem in (\ref{eqn:rate_sele_prob}) is equivalent to
\begin{equation}\label{eqn:type_ir_rate_ref}
\begin{array}{*{20}{l}}
{\mathop {{\rm{\min}}}\limits_R }&{\Lambda \left( R \right)}\\ 
{{\rm{subject\,to}}}&{{{\mathcal T_0}} < R \le \frac{{{{\mathcal T_0}}}}{{1 - \Delta }}}.
\end{array}
\end{equation}
Due to the complicated form of ${\Lambda \left( R \right)}$, it is difficult to derive a closed-form solution for the optimal transmission rate. Fortunately, after analyzing the function ${\Lambda \left( R \right)}$, we find a special property of ${\Lambda \left( R \right)}$ in the following.
%
%
\begin{property}\label{the:bounds_ee_ir}
The function ${\Lambda \left( R \right)}$ is bounded by
\begin{multline}\label{eqn:Lambda_lowerbound}
{\left( {\ln \left( 2 \right)} \right)^{\frac{{{2^{ - 1}} - {2^{ - L}}}}{{1 - {2^{ - L}}}}}}{\kappa _L}^{-\frac{1}{{1 - {2^{ - L}}}}}{\left( {\varpi \left( R \right)} \right)^{\frac{{{2^{ - 1}}}}{{1 - {2^{ - L}}}}}} \le \Lambda \left( R \right) \\
 \le
{\left( {\ln \left( 2 \right)} \right)^{\frac{{{2^{ - 1}} - {2^{ - L}}}}{{1 - {2^{ - L}}}}}}{\kappa _{L - 1}}^{-\frac{{{2^{ - 1}}}}{{1 - {2^{ - L}}}}}{\left( {\varpi \left( R \right)} \right)^{\frac{2^{-1}}{{1 - {2^{ - L}}}}}},
\end{multline}
where $\varpi \left( R \right) = {\left( {R - {{\cal T}_0}} \right)^{ - {2^{ - L + 1}}}}\left( {{2^R} - 1} \right)R$.
\end{property}
\begin{proof}
Please see Appendix \ref{app:proof_ratio_g_ineqs}.
\end{proof}
With these bounds and Intermediate Value Theorem \cite[Theorem 4.23]{rudin1964principles}, $\Lambda \left( R \right)$ can be rewritten as
\begin{equation}\label{eqn:Lambda_exc}
\Lambda \left( R \right) = {\left( {\ln \left( 2 \right)} \right)^{\frac{{{2^{ - 1}} - {2^{ - L}}}}{{1 - {2^{ - L}}}}}}\zeta {\left( {\varpi \left( R \right)} \right)^{\frac{{{2^{ - 1}}}}{{1 - {2^{ - L}}}}}},
\end{equation}
where $\zeta$ is bounded as
 \begin{equation}\label{eqn:zeta_bounds}
{\kappa _L}^{ - \frac{1}{{1 - {2^{ - L}}}}} \le \zeta  \le {\kappa _{L - 1}}^{ - \frac{{{2^{ - 1}}}}{{1 - {2^{ - L}}}}}.
 \end{equation}
Based on (\ref{eqn:Lambda_exc}), the optimization problem (\ref{eqn:type_ir_rate_ref}) can be rewritten as
\begin{equation}\label{eqn:eqn_opt}
\begin{array}{*{20}{l}}
{\mathop {{\rm{\min}}}\limits_R }&{\varpi \left( R \right) }\\
{{\rm{subject\,to}}}&{{{\mathcal T_0}} < R \le \frac{{{{\mathcal T_0}}}}{{1 - \Delta }}}.
\end{array}
\end{equation}

%
Following a similar proof as that for Theorem \ref{the:optimal_rate}, the optimal solution to (\ref{eqn:eqn_opt}) can be derived by using KKT conditions, as given in the following theorem.
\begin{theorem}\label{the:optimal_rate_ir}
The optimal rate to the problem in (\ref{eqn:eqn_opt}) is
\begin{equation}\label{eqn:opt_rate_type_ir_re}
{R^*} = \min \left\{ {\frac{{{{\mathcal T_0}}}}{{1 - \Delta }}},\Upsilon^{-1}(0) \right\}.
\end{equation}
where $\Upsilon \left( R \right) = \left( {R - {{\mathcal T_0}}} \right) \left( {{2^R}\ln \left( 2 \right)R + {2^R} - 1} \right)- {2^{ - L + 1}}\left( {{2^R} - 1} \right)R$ and $\Upsilon^{-1}(0)$ refers to the zero point of $\Upsilon(R)$. When $\varepsilon \ge 2^{-L}$, the optimal transmission rate reduces to $R^*=\Upsilon^{-1}(0)$.
\end{theorem}
The proof is omitted here to avoid redundancy. Notice that the zero point $\Upsilon^{-1}(0)$ can be easily found since $\Upsilon(R)$ is an increasing function with respect to $R$.

\subsubsection{Optimal Energy Efficiency}
With the definition of ${\Lambda \left( R \right)}$, the energy efficiency of HARQ-IR in (\ref{eqn:eta_ir_l_simp}) can be written as
\begin{equation}
\label{IR_1}
{\eta _{IR,L}} =
 \frac{{{\psi {\theta _L}}}{{{{\mathcal T_0}}}}}{{{{\left( {1 - \frac{{{{\mathcal T_0}}}}{R}} \right)}^{ - c}}{{\left( {\prod\limits_{k = 1}^L {{{\left( {\frac{{{g_k}\left( R \right)}}{{{g_{k - 1}}\left( R \right)}}} \right)}^{{2^{ - k}}}}} } \right)}^{\frac{1}{{1 - {2^{ - L}}}}}}}}
\end{equation}
As proved in Appendix \ref{app:g_k_thri_ine}, we have the following inequality
\begin{equation}\label{eqn:equa_ratio_upper_comp_cc_ir}
\frac{{{g_k}\left( R \right)}}{{{g_{k - 1}}\left( R \right)}} \le \frac{{{2^R} - 1}}{k}.
\end{equation}
Applying this inequality to (\ref{IR_1}) yields
\begin{align}
\label{IR_2}
{\eta _{IR,L}} &\ge
 \frac{{\left(\prod\limits_{k = 1}^L k^{2^{-k}}\right)^{\frac{1}{{1 - {2^{ - L}}}}}}{{\psi {\theta _L}}}{{{{\mathcal T_0}}}}{{\left( {1 - \frac{{{{\mathcal T_0}}}}{R}} \right)}^{ c}}}{{{{\left( {\prod\limits_{k = 1}^L {{{\left({{2^R} - 1} \right)}^{{2^{ - k}}}}} } \right)}^{\frac{1}{{1 - {2^{ - L}}}}}}}}\notag\\
 &= {\kappa _L}^{\frac{1}{{1 - {2^{ - L}}}}} {{\psi {\theta _L}}}{{{{\mathcal T_0}}}}\frac{{{{\left( {1 - \frac{{{{\mathcal T_0}}}}{R}} \right)}^c}}}{{{2^R} - 1}}= {\eta _{CC,L}}.
\end{align}
It means that HARQ-IR outperforms HARQ-CC in terms of energy efficiency. Moreover, the optimal energy efficiency of HARQ-IR is also higher than or equal to that of HARQ-CC, i.e., ${\eta _{IR,L}^*} \ge {\eta _{CC,L}^*}$.

Now putting the optimal rate (\ref{eqn:opt_rate_type_ir_re}) into (\ref{eqn:Lambda_exc}) and then (\ref{eqn:eta_ir_l_simp}) together with the bounds of $\zeta$ in (\ref{eqn:zeta_bounds}), the optimal energy efficiency of HARQ-IR is bounded as
\begin{equation}\label{eqn:energy_efficiency_ir_notbounds}
{\kappa _{L - 1}}^{\frac{{{2^{ - 1}}}}{{1 - {2^{ - L}}}}}\tilde \eta _{IR,L}^* \le \eta _{IR,L}^* \le {\kappa _L}^{\frac{1}{{1 - {2^{ - L}}}}} \tilde \eta _{IR,L}^*.
\end{equation}
where
\begin{equation}\label{eqn:energy_eff_tilde_def}
\tilde \eta _{IR,L}^* = {{\psi {\theta _L}}}{{{{\mathcal T_0}}}}{\left( {\ln \left( 2 \right)} \right)^{-\frac{{{2^{ - 1}} - {2^{ - L}}}}{{1 - {2^{ - L}}}}}}{\left( {\varpi \left(R ^*\right)} \right)^{-\frac{2^{-1}}{{1 - {2^{ - L}}}}}}.
\end{equation}
Combining (\ref{eqn:energy_efficiency_ir_notbounds}) with the inequality ${\eta _{IR,L}^*} \ge {\eta _{CC,L}^*}$, the optimal energy efficiency of HARQ-IR is consequently found to be bounded by
\begin{equation}\label{eqn:energy_efficiency_ir_bounds_renew}
\max\left\{ {\kappa _{L - 1}}^{\frac{{{2^{ - 1}}}}{{1 - {2^{ - L}}}}}\tilde \eta _{IR,L}^*, {\eta _{CC,L}^*}\right\} \le \eta _{IR,L}^* \le {\kappa _L}^{\frac{1}{{1 - {2^{ - L}}}}} \tilde \eta _{IR,L}^*.
\end{equation}
When the number of transmissions approaches infinity, i.e., $L \to \infty$, the following bounds also hold
\begin{equation}\label{eqn:energy_efficiency_ir_bounds_renew_infty}
\max\left\{ \sqrt{\kappa _{\infty}} \tilde \eta _{IR,\infty}^*, {\eta _{CC,\infty}^*}\right\} \le \eta _{IR,\infty}^* \le {\kappa _\infty} \tilde \eta _{IR,\infty}^*,
\end{equation}
where ${\eta _{CC,\infty}^*}$ has been given in Theorem \ref{the:ee_CC}, and $ \tilde \eta _{IR,\infty}^*$ is defined as $ \tilde \eta _{IR,\infty}^* \triangleq \mathop {\lim }\limits_{L \to \infty } \tilde \eta _{IR,L}^*$ and can be further derived from (\ref{eqn:energy_eff_tilde_def}) as
\begin{align}\label{eqn:energy_eff_tilde_IR}
 &\tilde \eta _{IR,\infty}^* \notag\\
 &= \mathop {\lim }\limits_{L \to \infty } {{\psi {\theta _L}}}{{{{\mathcal T_0}}}}{\left( {\ln \left( 2 \right)} \right)^{-\frac{{{2^{ - 1}} - {2^{ - L}}}}{{1 - {2^{ - L}}}}}}\notag\\
 &\quad\times {\left( {{{\left( {{{R}^*} - {{\mathcal T_0}}} \right)}^{ - {2^{ - L + 1}}}}\left( {{2^{{{R}^*}}} - 1} \right){{R}^*}} \right)^{-\frac{{{2^{ - 1}}}}{{1 - {2^{ - L}}}}}}\notag\\
& = \frac{{\theta _\infty}{\mathcal T_0}}{4\sqrt{\ln(2)}}
\mathop {\lim }\limits_{L \to \infty } {\left( {{{\left( {{{R}^*} - {{\mathcal T_0}}} \right)}^{ - {2^{ - L + 1}}}}\left( {{2^{{{R}^*}}} - 1} \right){{R}^*}} \right)^{-\frac{{{2^{ - 1}}}}{{1 - {2^{ - L}}}}}}.
\end{align}
In (\ref{eqn:energy_eff_tilde_IR}), the last equality follows based on the limit result $ \mathop {\lim }\limits_{L \to \infty } \psi= \frac{1}{4}$.

When $L \to \infty$, the inequality $\varepsilon \ge 2^{-L}$ would hold. Based on Theorem \ref{the:optimal_rate_ir}, the optimal transmission rate is $R^*=\Upsilon^{-1}(0)$ and thus $\Upsilon(R^*)=0$ which can be rewritten as
\begin{equation}\label{eqn:tilde_R_opt}
{{R}^*} = {{\mathcal T_0}} + \frac{{{2^{ - L + 1}}\left( {{2^{{{R}^*}}} - 1} \right){{R}^*}}}{{{2^{{{R}^*}}}\ln \left( 2 \right){{R}^*} + {2^{{{R}^*}}} - 1}}.
\end{equation}
Meanwhile, noticing that the optimal rate ${{R}^*}$ is bounded as ${\mathcal T_0} < {{R}^*} \le \frac{\mathcal T_0}{1-\Delta}$, $\Delta = \min\left\{\varepsilon, 2^{-L}\right\}$ and $\mathop {\lim }\limits_{L \to \infty } \Delta =0$, it follows by using squeeze theorem that $\mathop {\lim }\limits_{L \to \infty } {{R}^*} = {\mathcal T_0}$. With this limit and plugging (\ref{eqn:tilde_R_opt}) into (\ref{eqn:energy_eff_tilde_IR}), it yields
\begin{align}\label{eqn:eta_tilde_fina_opt}
&\tilde \eta _{IR,\infty }^*\notag\\ 
 &=\frac{{{\theta _\infty }{{\cal T}_0}}}{{4\sqrt {\ln (2)} }}\mathop {\lim }\limits_{L \to \infty } {\left( {\left( {{2^{{{R}^*}}} - 1} \right){{R}^*}} \right)^{ - \frac{{{2^{ - 1}}}}{{1 - {2^{ - L}}}}}}\mathop {\lim }\limits_{L \to \infty } {\left( {{{R}^*} - {{\cal T}_0}} \right)^{\frac{{{2^{ - L}}}}{{1 - {2^{ - L}}}}}}\notag\\
 &=\frac{{{\theta _\infty }}}{4}\sqrt {\frac{{{\mathcal T_0}}}{{\ln (2)\left( {{2^{{\mathcal T_0}}} - 1} \right)}}} \mathop {\lim }\limits_{L \to \infty } {\left( {\frac{{{2^{ - L + 1}}\left( {{2^{{{R}^*}}} - 1} \right){{R}^*}}}{{{2^{{{R}^*}}}\ln \left( 2 \right){{R}^*} + {2^{{{R}^*}}} - 1}}} \right)^{\frac{{{2^{ - L}}}}{{1 - {2^{ - L}}}}}}\notag\\
&= \frac{{{\theta _\infty }}}{4}\sqrt {\frac{{{{\cal T}_0}}}{{\ln (2)\left( {{2^{{{\cal T}_0}}} - 1} \right)}}}.
\end{align}
Plugging (\ref{eqn:upper_bound_cc}) and (\ref{eqn:eta_tilde_fina_opt}) into (\ref{eqn:energy_efficiency_ir_bounds_renew_infty}), it then follows that
\begin{multline}\label{eqn:final_harq_ir_eff}
\max \left\{ {\frac{{{\theta _\infty }}}{4}\sqrt {\frac{{{\kappa _\infty }{{\cal T}_0}}}{{\ln (2)\left( {{2^{{{\cal T}_0}}} - 1} \right)}}} ,\frac{{{\kappa _\infty }{\theta _\infty }{{\cal T}_0}}}{{4\left( {{2^{{{\cal T}_0}}} - 1} \right)}}} \right\} \le \eta _{IR,\infty }^* \\
\le \frac{{{\kappa _\infty }{\theta _\infty }}}{4}\sqrt {\frac{{{{\cal T}_0}}}{{\ln (2)\left( {{2^{{{\cal T}_0}}} - 1} \right)}}}.
\end{multline}
Clearly, both the lower and upper bounds of $\eta _{IR,\infty }^*$ in (\ref{eqn:final_harq_ir_eff}) are decreasing functions of the the goodput threshold ${\cal T}_0$. In addition, the following inequality of their limits holds
\begin{multline}\label{eqn:energy_eff_harq_bound_limi}
\max \left\{ {\frac{{\theta _\infty }}{4}\sqrt {\frac{{{\kappa _\infty }}}{{\ln (2)}}} \sqrt {\mathop {\lim }\limits_{{{\mathcal T}_0} \to 0} \frac{{{\mathcal T_0}}}{{{2^{{\mathcal T_0}}} - 1}}} ,\frac{{{\kappa _\infty }}{\theta _\infty }}{4}\mathop {\lim }\limits_{{{\cal T}_0} \to 0} \frac{{{\mathcal T_0}}}{{{2^{{\mathcal T_0}}} - 1}}} \right\} \\ \le \mathop {\lim }\limits_{{{\cal T}_0} \to 0} \eta _{IR,\infty }^*
\le \frac{{{\kappa _\infty }}{\theta _\infty }}{{4\sqrt {\ln (2)} }}\sqrt {\mathop {\lim }\limits_{{{\cal T}_0} \to 0} \frac{{{\mathcal T_0}}}{{{2^{{\mathcal T_0}}} - 1}}}.
\end{multline}
Noticing that $\mathop {\lim }\limits_{{\mathcal T_0} \to 0} \frac{{{\mathcal T_0}}}{{{2^{{\mathcal T_0}}} - 1}} = \frac{1}{\ln \left( 2 \right)}$, we have
\begin{equation}\label{eqn:energy_efficin_ir_boun_lim_rew}
\max \left\{ {\frac{{\sqrt {{\kappa _\infty }} }{\theta _\infty }}{{4\ln (2)}},\frac{{{\kappa _\infty }}{\theta _\infty }}{{4\ln (2)}}} \right\} \le \mathop {\lim }\limits_{{{\cal T}_0} \to 0} \eta _{IR,\infty }^* \le \frac{{{\kappa _\infty }}{\theta _\infty }}{{4\ln (2)}}.
\end{equation}
Since ${\kappa _\infty } > 1$ and using squeeze theorem, it follows
\begin{equation}
\label{IR_limit}
\mathop {\lim }\limits_{{{\cal T}_0} \to 0} \eta _{IR,\infty }^* = \frac{{{\kappa _\infty }}{\theta _\infty }}{{4\ln (2)}}.
\end{equation}
Summarizing the results in (\ref{eqn:final_harq_ir_eff}), (\ref{IR_limit}) and Property \ref{the:type_I_eff} finally leads to the following result of the optimal energy efficiency of HARQ-IR.
\begin{theorem}\label{the:ee_IR}
Under time-correlated Rayleigh fading channels, the optimal energy efficiency of HARQ-IR with outage and goodput constraints is upper bounded by
\begin{equation}\label{eqn:upper_bound_ir}
{\eta _{IR,L}^*} \le {\eta _{IR,\infty}^*} \le \frac{{{\kappa _\infty }{\theta _\infty }}}{4}\sqrt {\frac{{{{\cal T}_0}}}{{\ln (2)\left( {{2^{{{\cal T}_0}}} - 1} \right)}}}.
\end{equation}
Particularly, for Rayleigh fading channels with unit channel gains ${\sigma_1}^2=\cdots={\sigma_L}^2=1$, $\theta_\infty \le 1$ and the maximal energy efficiency of HARQ-IR with outage and goodput constraints follows
\begin{equation}
{\eta _{IR,\infty}^*} \le \frac{{{\kappa _\infty }{\theta _\infty }}}{4}\sqrt {\frac{{{{\cal T}_0}}}{{\ln (2)\left( {{2^{{{\cal T}_0}}} - 1} \right)}}} \le  \frac{\kappa_\infty}{4\ln(2)} \triangleq {\eta _{IR,\infty}^{\max}}.
\end{equation}
%
\end{theorem}
Comparing the results in Theorems \ref{the:ee_CC} and \ref{the:ee_IR}, we can conclude that HARQ-CC and HARQ-IR can reach the same maximal energy efficiency of $\frac{\kappa_\infty}{4\ln(2)}$.

\section{Numerical Results and Discussions}\label{sec:numerical}
The performance of our proposed energy-efficient optimization is tested and the impact of various parameters is discussed in this section. Unless otherwise stated, the results are provided over time-correlated Rayleigh fading channels with time correlation $\rho=0.5$ and unit channel gains ${\sigma_1}^2=\cdots={\sigma_L}^2=1$.

\subsection{Numerical Verification}
In our design, asymptotic outage probability is adopted to enable the derivation of closed-form solutions for the optimal transmission powers and the optimal rate. To validate the correctness of our design, the optimal energy efficiency achieved by our design is compared with that achieved through numerical exhaustive search based on the exact outage probabilities derived in \cite{shi2015outage,shi2016optimal,shi2017asymptotic}. The results versus the outage tolerance $\varepsilon$ and the minimum goodput requirement $\mathcal T_0$  are shown in Figs. \ref{fig:validationEpsilon} and \ref{fig:validation}, respectively. It is readily observed from Fig. \ref{fig:validationEpsilon} that there is an excellent match between the optimal energy efficiencies based on the asymptotic and the exact outage probabilities when $\varepsilon \le 10^{-2}$, because the outage probability can be well approximated by the asymptotic outage probability (\ref{eqn:outage_prob}) under a \emph{low outage} (or \emph{high SNR}). This result demonstrates that the effectiveness of the proposed close-form solution holds true when outage constraint is strict, i.e., the allowable outage probability is small, which is usually true in practical applications since the high QoS is generally required in practice. The effectiveness of our design is further demonstrated by the results in Fig. \ref{fig:validation} where the optimal energy efficiency achieved by our design coincides well with that achieved based on the exact outage probability no matter how large the minimum goodput. Moreover, it is also shown in Fig. \ref{fig:validationEpsilon} that the optimal energy efficiency $\eta_L^*$ is an increasing function of $\varepsilon$. For instance, $\eta_L^*$ of HARQ-IR scheme is increased by about $0.15$ bits/J when the outage tolerance is relaxed from $10^{-3}$ to $10^{-1}$. However, as shown in Fig. \ref{fig:validation}, the optimal energy efficiency decreases with the increase of the minimum goodput requirement $\mathcal T_0$, which verifies our result that the maximal energy efficiency is achieved when no goodput constraint is considered, i.e., $\mathcal T_0 \to 0$. In addition, both Figs. \ref{fig:validationEpsilon} and \ref{fig:validation} show that HARQ-IR performs the best in terms of optimal energy efficiency, while Type I HARQ provides the worst performance without exploiting the erroneously received sub-codewords.
\begin{figure}
  \centering
  \includegraphics[width=2.35in]{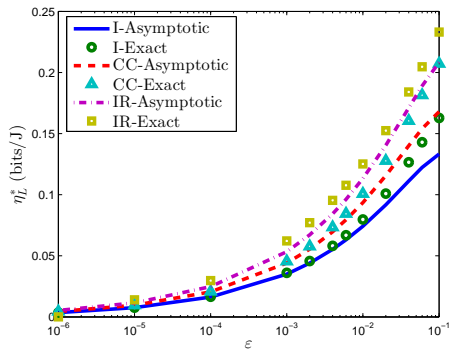}
  \caption{The optimal energy efficiency versus the outage tolerance with $L=2$ and $\mathcal T _0 = 2$ bps/Hz.}\label{fig:validationEpsilon}
\end{figure}

\begin{figure}
  \centering
  \includegraphics[width=2.35in]{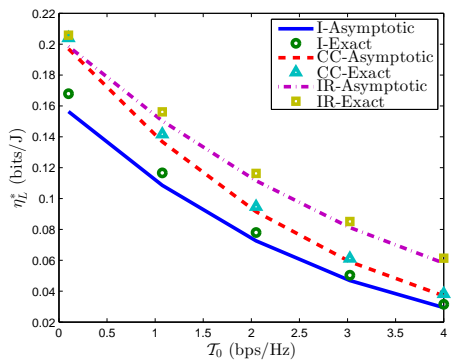}
  \caption{The optimal energy efficiency versus the minimum goodput requirement with $L=2$ and $\varepsilon=10^{-2}$.}\label{fig:validation}
\end{figure}

To further demonstrate the superiority of our design, our energy-efficient design is compared with uniform power design and their optimal energy efficiencies are shown in Fig. \ref{fig:percompare}. Notice that in the uniform power design, the transmission powers at different HARQ rounds are set equal, i.e., $P_1=\cdots=P_L=P$, and the transmission power $P$ and the rate are then optimally chosen to maximize the energy efficiency under the same QoS constraints as our design. Clearly, our design can achieve a significant enhancement of energy efficiency over the uniform power design under a stringent outage constraint. Interestingly, for Type I HARQ, the energy efficiency achieved by the proposed design converges to that of the uniform power design under loose outage constraint, i.e., $\varepsilon \to 1$. However, there is a non-negligible gap between the energy efficiencies of the proposed design and the uniform power design for both HARQ-CC and HARQ-IR. This is due to the fact that the erroneously received packets are directly discarded by Type I HARQ and thus the previously consumed resources such as transmission powers are not exploited. To overcome this problem, both HARQ-CC and HARQ-IR combine the currently received packet with the erroneously received packets for reutilizing these resources. Particularly for HARQ-CC and HARQ-IR, the concluded results are totally different from \cite{tajan2013hybrid}, where the uniform power allocation can offer similar performance as the optimal solution for spectral efficiency maximization. This essentially stems from the difference between the energy and the spectral efficiencies. The energy efficiency is in fact a ratio of spectral efficiency to the long term average power \cite{tuninetti2011benefits}, so it would be significantly affected by transmission powers. As opposed to \cite[Fig. 3]{jabi2016energy}, \cite{tajan2013hybrid} where the spectral efficiency maximization is targeted and only a slight spectral efficiency improvement can be achieved via optimal power allocation, the adaptation of transmission powers appears to be very crucial to the energy efficiency maximization.


\begin{figure}
  \centering
  \includegraphics[width=2.35in]{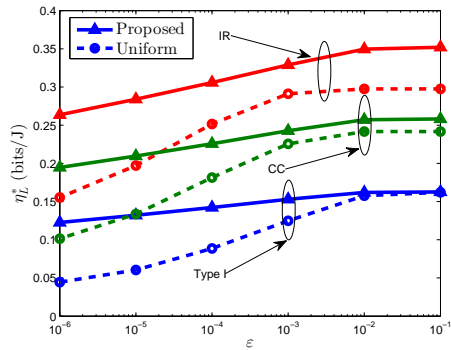}
  \caption{Comparison of the proposed design with uniform power design for various HARQ schemes by setting $\mathcal T_0 = 2$ bps/Hz and $L=5$.}\label{fig:percompare}
\end{figure}

\subsection{Impact of System Parameters}
To investigate the impact of various system parameters, the optimal energy efficiency versus the maximal number of transmissions is shown in Fig. \ref{fig:eevL}. It is clear that the optimal energy efficiency $\eta_L^*$ is an increase function of $L$ and is upper bounded. The maximal energy efficiency is achieved when $L \to \infty $, which is named as lossless HARQ for convenience. This is consistent with our analytical results in Theorems \ref{the:ee_I}, \ref{the:ee_CC} and \ref{the:ee_IR}. Since the energy efficiency is upper bounded, when $L$ is large, further increase of the maximal number of transmission can only contribute a very limited improvement on the energy efficiency, but it would definitely lead to the reduction of spectral efficiency since more transmissions are conducted to convey the same information. Thus there would exist a tradeoff between the energy efficiency and the spectral efficiency, which will be particularly discussed later. Herein, it should be pointed out  that $\theta_\infty$ is a function of $\rho$, and can be approximated as $\theta_L$ with high accuracy by choosing $L=20$.
\begin{figure}
  \centering
  \includegraphics[width=2.35in]{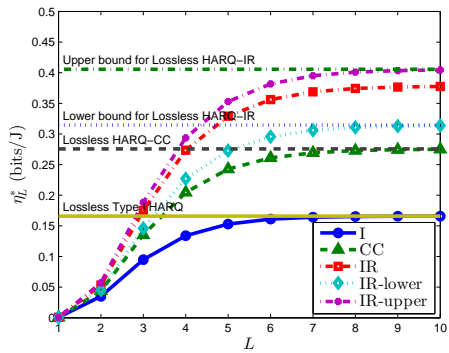}
  \caption{The optimal energy efficiency versus the maximal number of transmissions with $\mathcal T_0 = 2$ bps/Hz and $\varepsilon=10^{-4}$.}\label{fig:eevL}
\end{figure}

To test the effect of the outage tolerance $\varepsilon$, the optimal energy efficiency $\eta_L^*$ versus $\varepsilon$ is plotted in Fig. \ref{fig:eevoutt}. It is found that the optimal energy efficiency first increases as $\varepsilon$ increases up to $0.06$, while it remains constant when $\varepsilon$ increases further. This result can be well explained by Theorem \ref{the:optimal_rate} and Theorem \ref{the:optimal_rate_ir}. More precisely, when $\varepsilon > 2^{-L}$, the optimal rate is independent of the outage tolerance $\varepsilon$ and thus the optimal energy efficiency becomes irrelevant to $\varepsilon$. In the case of $L=4$, we have $2^{-L} = 2^{-4} \approx 0.06$. Therefore, when $\varepsilon > 0.06$, the optimal energy efficiency would become constant in this case.
\begin{figure}
  \centering
  \includegraphics[width=2.35in]{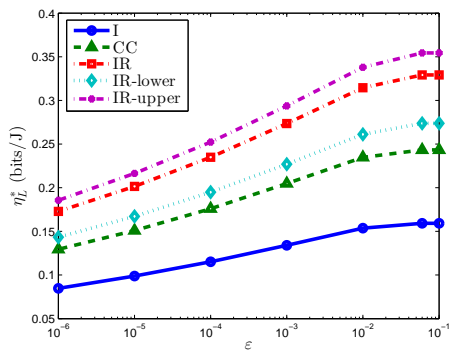}
  \caption{The optimal energy efficiency versus the outage tolerance with  $\mathcal T_0 = 2$ bps/Hz and $L=4$.}\label{fig:eevoutt}
\end{figure}

The effect of the minimal goodput requirement $\mathcal T_0$ is then shown in Fig. \ref{fig:eevr}. It can be seen that the optimal energy efficiency $\eta_L^*$ decreases with the increase of $\mathcal T_0$. For all the three HARQ schemes, the maximal energy efficiency is achieved when $\mathcal T_0 \to 0$. Moreover, HARQ-CC and HARQ-IR can achieve the same maximal energy efficiency of $\frac{\kappa_\infty}{4\ln(2)}$ when $\mathcal T_0 \to 0$. However, the superiority of HARQ-IR over HARQ-CC becomes more significant as $\mathcal T_0$ increases.

\begin{figure}
  \centering
  \includegraphics[width=2.35in]{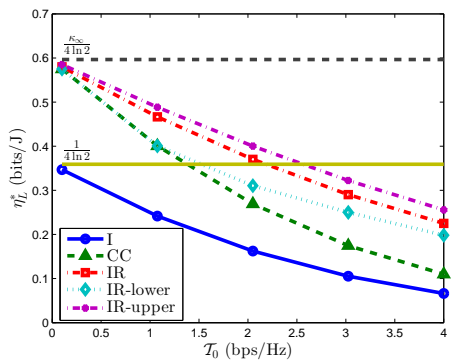}
  \caption{The optimal energy efficiency versus the minimum goodput requirement with $L=10$ and $\varepsilon=10^{-2}$.}\label{fig:eevr}
\end{figure}

The impact of the time correlation $\rho$ on the optimal energy efficiency $\eta_L^*$ is finally investigated and the results are shown in Fig. \ref{fig:timecorr}. As proved in Appendix \ref{app:time_corr}, $\theta_L$ is a decreasing function of the time correlation $\rho$ and thus the time correlation has a detrimental impact on the energy efficiency. This can be verified by the results in Fig. \ref{fig:timecorr}. It can also been seen that there is a significant drop of the energy efficiency when $\rho > 0.5$, which is consistent with the result in \cite{goldsmith2005wireless} that time correlation lower than $0.5$ would not lead to a significant performance degradation.
\begin{figure}
  \centering
  \includegraphics[width=2.35in]{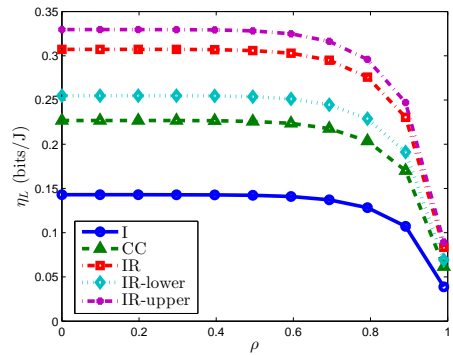}
  \caption{Effect of time correlation with $L=5$, $\mathcal T_0=2$ bps/Hz and $\varepsilon=10^{-4}$.}\label{fig:timecorr}
\end{figure}
\subsection{Spectral Efficiency}
As defined in \cite{wu2014energy,chelli2014performance}, spectral efficiency of HARQ schemes signifies the average number of successfully delivered bits per channel use (bps/Hz) and is defined as \begin{equation}\label{eqn:spect_eff}
{\xi _L} = \frac{{R(1 - {p_{out,L}})}}{{\sum\nolimits_{l = 0}^{L - 1} {{p_{out,l}}} }}.
\end{equation}
Generally, the optimizations of the energy efficiency and the spectral efficiency are two conflicting goals. This conflict can be clearly observed from the results of energy and spectral efficiencies achieved by our design as shown in Fig. \ref{fig:se}. It is shown that the spectral efficiency ${\xi _L}$ decreases while the optimal energy efficiency $\eta_L^*$ increases with the increase of $L$ and higher energy efficiency is achieved at the cost of spectral efficiency degradation. In addition, with our energy-efficient optimization, the HARQ-IR scheme can achieve the highest energy efficiency but with the lowest spectral efficiency since our design is targeting at energy efficiency maximization. Notice that this result is different from that in \cite{chelli2013performance,larsson2014throughput,caire2001throughput} where the design objective is to maximize the spectral efficiency. To further illustrate the tradeoff between the energy and spectral efficiencies for the three HARQ schemes, the minimum goodput requirement $\mathcal T_0$ is varied from $0.5 $bps/Hz to $10$bps/Hz and the corresponding efficiencies are shown in Fig. \ref{fig:sevsee}. 
It is clear from Figs. \ref{fig:se} and \ref{fig:sevsee} that HARQ-CC can always achieve spectral and energy efficiencies in between those of Type I HARQ and HARQ-IR given the same objective and constraints for optimization. In other words, HARQ-CC can achieve a better tradeoff between the energy and spectral efficiencies than the other schemes given the same objective and constraints for optimization.

%

\begin{figure}
  \centering
  \includegraphics[width=2.35in]{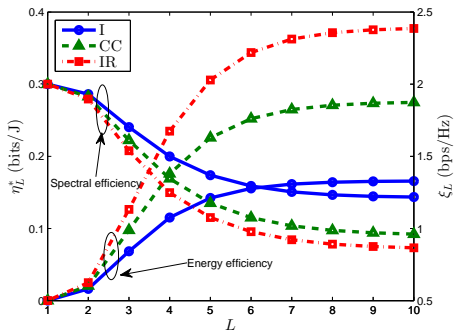}
  \caption{Spectral and energy efficiencies of various HARQ schemes with $\mathcal T_0 = 2$ bps/Hz and $\varepsilon=10^{-4}$.}\label{fig:se}
\end{figure}

\begin{figure}
  \centering
  \includegraphics[width=2.35in]{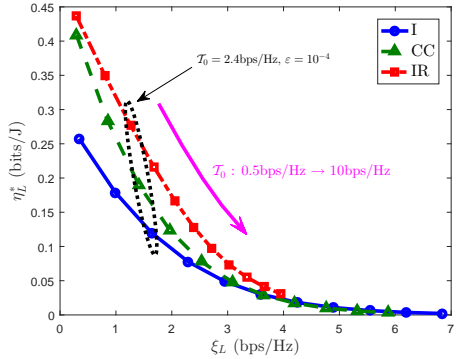}
  \caption{{Tradeoff between the optimal energy efficiency and spectral efficiency with $\varepsilon=10^{-4}$ and $L=5$.}}\label{fig:sevsee}
\end{figure}

To demonstrate the generality of the above results, a practical long term evolution (LTE) system with a coding rate of $1/2$ and a modulation scheme of 16QAM is also taken as an example for simulations. The system level simulation is conducted by using LTE system toolbox for MATLAB. Note that since ${{{\mathcal T_0}} < R \le \frac{{{{\mathcal T_0}}}}{{1 - \Delta }}}$ and $\Delta$ is small, ${\mathcal T_0}$ can be approximated as ${\mathcal T_0} \approx R = 1/2 \times \log_2 16 = 2$ bps/Hz. Hereby, by using the closed-form solution of problem (\ref{eqn:joint_adaption}) with ${\mathcal T_0}=2$ bps/Hz, Fig. \ref{fig:sevseereal} shows the spectral and the energy efficiencies of the three HARQ schemes against the outage constraint under LTE system settings. Similar results can be observed as Fig. \ref{fig:se}, that is, HARQ-CC can strike the best balance between the energy and the spectral efficiencies.



\begin{figure}
  \centering
  \includegraphics[width=2.35in]{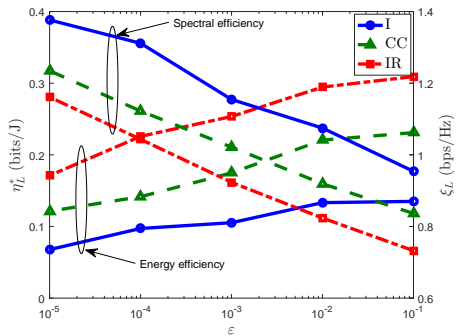}
  \caption{Spectral and energy efficiencies for the three HARQ schemes in LTE system with a coding rate of $1/2$, a modulation scheme of 16QAM, $N_b=12960$ bits and $L=4$.}\label{fig:sevseereal}
\end{figure}

\section{Conclusion}\label{sec:con}
Energy-efficient optimization for HARQ schemes has been proposed in this paper. Different from the prior designs, widely occurred time-correlated fading channels and practical QoS constraints are considered in the optimal design to maximize the energy efficiency. The optimal transmission powers and the optimal rate have been derived in closed-forms and have further enabled the analysis of the maximal energy efficiencies of various HARQ schemes. It has been found that with low outage constraint, the maximal energy efficiencies of Type I HARQ and HARQ-CC/IR are $\frac{1}{4\ln2}$ bits/J and $\frac{\kappa_\infty}{4\ln2}$ bits/J, respectively. Our numerical results have also shown that the energy efficiency improvement is achieved in sacrifice of the spectral efficiency, because the spectral efficiency and the energy efficiency are two conflicting objectives. In addition, HARQ-CC can achieve a better tradeoff between the energy efficiency and the spectral efficiency than Type I HARQ and HARQ-IR.


\appendices
\section{Proof of Theorem \ref{the:optimal_outage}}\label{app:proof_tarout}
It is clear from (\ref{eqn:P_L_star_sec}) that the target outage probability $\alpha$ should be nonzero, i.e., $\alpha \ne 0$. Together with the first constraint in (\ref{eqn:tar_out_prob}), we have $0 < \alpha \le \varepsilon$. Further considering the second constraint $R(1- \alpha) \ge \mathcal T_0$, the transmission rate should satisfy $R > \mathcal T_0$, otherwise the second constraint cannot be satisfied and the problem (\ref{eqn:tar_out_prob}) is then infeasible. Then the problem (\ref{eqn:tar_out_prob}) can be reformulated as
\begin{equation}\label{eqn:tar_out_prob_re}
\begin{array}{*{20}{l}}
{\mathop {{\rm{max}}}\limits_{\alpha } }&{f\left( \alpha  \right)}\\
{{\rm{subject\,to}}}&{0 < \alpha  \le \min\left\{\varepsilon,1 - \frac{{{\mathcal T_0}}}{R}\right\} }.
\end{array}
\end{equation}

The Lagrangian associated with (\ref{eqn:tar_out_prob_re}) can be written as
\begin{equation}\label{eqn:lagrangain_2}
{\mathcal L_1}\left( {\alpha,u} \right) = f(\alpha) + u\left( {\alpha  - \min\left\{\varepsilon,1 - \frac{{{\mathcal T_0}}}{R}\right\} } \right)-\nu\alpha,
\end{equation}
where $u$ and $\nu$ refer to Lagrange multipliers. The optimal solution to (\ref{eqn:tar_out_prob_re}) should satisfy the Karush-Khun-Tucker (KKT) conditions as
\begin{equation}\label{eqn:l2_kkt1}
{\left. {\frac{{\partial{{\cal L}_1}}}{{\partial\alpha }}} \right|_{\left( {{\alpha ^*},{u^*},\nu^*} \right)}} = f'({\alpha ^*}) + {u^*}-\nu^* = 0,
\end{equation}
\begin{equation}\label{eqn:l2_kkt2}
{u^*}\left( {{\alpha ^*} - \min \left\{ {\varepsilon ,1 - \frac{{{\mathcal T_0}}}{R}} \right\}} \right) = 0,
\end{equation}
\begin{equation}\label{eqn:l2_kkt21}
\nu^*\alpha^* = 0,
\end{equation}
$0 < {\alpha ^*}  \le \min \left\{ {\varepsilon ,1 - \frac{{{\mathcal T_0}}}{R}} \right\}$ and ${u^*}, \nu^* \ge 0$. Since $\alpha^* > 0$, (\ref{eqn:l2_kkt21}) implies $\nu^* = 0$.
Substituting $\nu^* = 0$ into (\ref{eqn:l2_kkt1}) yields
\begin{align}\label{eqn:u_star}
{u^*} &=  - f'({\alpha ^*}) =\frac{{{{\alpha^*} ^{\frac{1}{{{2^L} - 1}}}}}}{{{2^L} - 1}}\left({{{{\alpha^*} ^{ - 1}}-{2^L}}}\right).
\end{align}
Note that the Lagrange multiplier ${{u^*} } $ is either larger than or equal to $0$. Suppose that ${{u^*} } = 0$, it follows from (\ref{eqn:u_star}) that $\alpha^*  = {2^{ - L}}$. Combining this with (\ref{eqn:l2_kkt2}) and the constraint $0 < {\alpha ^*}  \le \min \left\{ {\varepsilon ,1 - \frac{{{\mathcal T_0}}}{R}} \right\}$, we have $\alpha^*  = {2^{ - L}} \le \min \left\{ {\varepsilon ,1 - \frac{{{{\mathcal T_0}}}}{R}} \right\}$. If ${{u^*} } > 0$, on the other hand, (\ref{eqn:l2_kkt2}) implies that ${{\alpha ^*} = \min \left\{ {\varepsilon ,1 - \frac{{{{\mathcal T_0}}}}{R}} \right\}} $. Besides, combining ${u^*} > 0$ with (\ref{eqn:u_star}) gives ${{\alpha^*} ^{ - 1}}-{2^L}  > 0$ so that ${{\alpha ^*} = \min \left\{ {\varepsilon ,1 - \frac{{{{\mathcal T_0}}}}{R}} \right\}} < 2 ^{-L}$. To summarize, the optimal target outage probability is therefore given by ${{\alpha ^*} = \min \left\{ {\varepsilon ,1 - \frac{{{{\mathcal T_0}}}}{R},2 ^{-L}} \right\}}$ for $R > \mathcal T_0$.

By defining $\Delta  = \min \left\{ {\varepsilon ,{2^{ - L}}} \right\}$,
the optimal target outage probability ${\alpha ^*}$ can be rewritten as
\begin{multline}\label{eqn:tager_out1}
{\alpha ^*} = \min \left\{ {\Delta ,1 - \frac{{{{\mathcal T_0}}}}{R}} \right\}= \left( {1 - \frac{{{{\mathcal T_0}}}}{R}} \right)\times\\
  \left( {\chi\left( {R - {{\mathcal T_0}}} \right) - \chi\left( {R - \frac{{{{\mathcal T_0}}}}{{1 - \Delta }}} \right)} \right) + \Delta \chi\left( {R - \frac{{{{\mathcal T_0}}}}{{1 - \Delta }}} \right),
\end{multline}
where $\chi (\cdot)$ denotes the indicator function as (\ref{eqn:indi_fundef}). Plugging (\ref{eqn:tager_out1}) into $f(\alpha)={{(1 - {\alpha}){\alpha}^{\frac{1}{{{2^L} - 1}}}}}$ finally leads to (\ref{eqn:f_alpha_star}).

\section{Proof of Theorem \ref{the:optimal_rate}}\label{app:proof_convex}
The Lagrangian associated with (\ref{eqn:type_i_rate_ref1}) is written as
\begin{equation}\label{eqn:lagrangain_gene_R}
{\mathcal L_2}\left( {R,v,w} \right) = \Phi \left( R \right) + v\left( {{{\mathcal T_0}} - R} \right) + w\left( {R - \frac{{{{\mathcal T_0}}}}{{1 - \Delta }}} \right),
\end{equation}
where $v$ and $w$ represent the Lagrange multipliers. The optimal solution to  (\ref{eqn:type_i_rate_ref1}) should satisfy the KKT conditions as
\begin{equation}\label{eqn:kkt1_cond_case3}
{\left. {\frac{{\partial {\mathcal L_2}}}{{\partial R}}} \right|_{\left( {{R^*},{v^*},{w^*}} \right)}} = \Phi '\left( R^* \right) - v^* + w^* = 0,
\end{equation}
\begin{equation}\label{eqn:kkt2_cond_case3}
v^*\left( {{{\mathcal T_0}} - R^*} \right) = 0,
\end{equation}
\begin{equation}\label{eqn:kkt3_cond_case3}
w^*\left( {R^* - \frac{{{{\mathcal T_0}}}}{{1 - \Delta }}} \right) = 0,
\end{equation}
${{{\mathcal T_0}} < R^* \le \frac{{{{\mathcal T_0}}}}{{1 - \Delta }}}$ and $v^*,w^* \ge 0$. Herein, the first derivative of $\Phi \left( R \right)$ with respect to $R$ is
\begin{align}\label{eqn:first_der_phi}
\Phi '\left( R \right) 
 &= \frac{1}{{{R^{1 - c}}{{\left( {R - {{\mathcal T_0}}} \right)}^{c + 1}}}}\varphi \left( R \right),
\end{align}
where $\varphi \left( R \right) = \ln \left( 2 \right)R\left( {R - {{\mathcal T_0}}} \right){2^R} - c{{\mathcal T_0}}\left( {{2^R} - 1} \right)$. Noticing ${{{\mathcal T_0}} } < R^*$, it follows from (\ref{eqn:kkt2_cond_case3}) that $v^*=0$. Clearly from (\ref{eqn:kkt1_cond_case3}), $w^* $ is thus given by
\begin{equation}\label{eqn:w_final}
w^* =  - \Phi '\left( R ^*\right) =  - \frac{1}{{{{R^*}^{1 - c}}{{\left( {R^* - {{\mathcal T_0}}} \right)}^{c + 1}}}}\varphi \left( R^* \right).
\end{equation}

Note that $w^*$ is either larger than or equal to $0$. If $w^*>0$, it follows from (\ref{eqn:kkt3_cond_case3}) that ${R^*}=\frac{{{{\mathcal T_0}}}}{{1 - \Delta }}$. Together with (\ref{eqn:w_final}), we have $\varphi \left( R^* \right) = \varphi \left( {\frac{{{{\mathcal T_0}}}}{{1 - \Delta }}} \right) < 0$. 
On the other hand, if $w^* = 0$, following from (\ref{eqn:w_final}), we have $\varphi \left( R ^* \right) = 0$. By defining $\varphi^{-1} \left( y \right)$ as the inverse function of $\varphi$, we have ${R^*}=\varphi^{-1}(0)$. Since $R > {\mathcal T_0}$, the following inequality holds
\begin{equation}\label{eqn:varphi_def}
\varphi '\left( R \right) = \ln \left( 2 \right){2^R}\left( {\ln \left( 2 \right)R\left( {R - {{\mathcal T_0}}} \right) + 2R - \left( {1 + c} \right){{\mathcal T_0}}} \right) > 0.
\end{equation}
It means that $\varphi \left( R \right)$ is a monotonically increasing function of the rate $R$. Together with $R^* \le \frac{{{{\mathcal T_0}}}}{{1 - \Delta }}$, we have $\varphi \left( \frac{{{{\mathcal T_0}}}}{{1 - \Delta }} \right) \ge \varphi \left( R ^* \right) = 0$. Then the optimal solution $R^*$ under the above two cases can be summarized as
\begin{equation}\label{eqn:opt_rate_type_i1}
{R^*} = \left\{ {\begin{array}{*{20}{c}}
{\frac{{{{\mathcal T_0}}}}{{1 - \Delta }},}&{\varphi \left( {\frac{{{{\mathcal T_0}}}}{{1 - \Delta }}} \right) < 0}\\
{{\varphi^{-1}(0)},}&{\varphi \left( {\frac{{{{\mathcal T_0}}}}{{1 - \Delta }}} \right) \ge 0.}
\end{array}} \right.
\end{equation}
Noticing that $\varphi \left( R \right)$ is an increasing function of $R$, (\ref{eqn:opt_rate_type_i1}) can be rewritten as (\ref{eqn:opt_rate_type_i_re}).

Clearly from (\ref{eqn:opt_rate_type_i1}), the optimal rate $R^*$ is determined by the sign of $\varphi \left( {\frac{{{{\mathcal T_0}}}}{{1 - \Delta }}} \right)$. Based on its definition, $\varphi \left( {\frac{{{{\mathcal T_0}}}}{{1 - \Delta }}} \right)$ can be explicitly expressed as
\begin{align}\label{eqn:varphi_app_smallR0}
\varphi \left( {\frac{{{{\mathcal T_0}}}}{{1 - \Delta }}} \right)
&= {2^{\frac{{{{\mathcal T_0}}}}{{1 - \Delta }}}}{{\mathcal T_0}}\left( { {\ln \left( 2 \right)\frac{{\Delta {{\mathcal T_0}}}}{{{{\left( {1 - \Delta } \right)}^2}}} + c\left( {{2^{ - \frac{{{{\mathcal T_0}}}}{{1 - \Delta }}}} - 1} \right)}} \right).
\end{align}
By using the inequality ${\rm 2}^{-x} \ge 1-x\ln2$, we have
\begin{equation}\label{eqn:phi_inelower}
\varphi \left( {\frac{{{{\cal T}_0}}}{{1 - \Delta }}} \right) \ge {2^{\frac{{{{\cal T}_0}}}{{1 - \Delta }}}}{{\cal T}_0}\ln \left( 2 \right)\frac{{{{\cal T}_0}}}{{1 - \Delta }}\left( {\frac{\Delta }{{1 - \Delta }} - c} \right).
\end{equation}
When $\varepsilon \ge 2^{-L}$, $\Delta=\min\{\varepsilon,2^{-L}\}=2^{-L}$ and the right hand side of (\ref{eqn:phi_inelower}) equals to 0 since $c = \frac{1}{{{2^L} - 1}}$. In other words, when $\varepsilon \ge 2^{-L}$, we have $\varphi \left( {\frac{{{{\cal T}_0}}}{{1 - \Delta }}} \right) \ge 0$. From (\ref{eqn:opt_rate_type_i1}), the optimal transmission rate is thus given by $R^*=\varphi^{-1}(0)$. The proof is then completed.

\section{Decreasing Monotonicity of $\theta_L$}\label{app:time_corr}
From (\ref{eqn:power_eff_opt}), the optimal energy efficiency can be rewritten as
\begin{equation}\label{eqn:power_eff_opt_app}
{{\eta}  _{I,L}^*}  = \mathcal K {\theta _L}.
\end{equation}
where $\mathcal K = {{\psi }}{{{{\mathcal T_0}}}} \frac{{{{\left( {{R^*} - {{\mathcal T_0}}} \right)}^c}}}{{{R^*}^c\left( {{2^{{R^*}}} - 1} \right)}}$ and is independent of $\rho$. Accordingly, the monotonicity of $\theta_L$ with respect to $\rho$ is the same as that of ${\eta _{I,L}^*}$. With the same monotonicity, we can prove the decreasing monotonicity of $\theta_L$ through the analysis of the monotonicity of ${\eta _{I,L}^*}$ as follows.

Specifically, as proved in \cite[Lemma 3]{shi2017asymptotic}, the asymptotic outage probability is an increasing function of $\rho$. Thus it follows from (\ref{eqn:power_effi}) that the energy efficiency $\eta_{I,L}$ is a decreasing function of $\rho$. To proceed with the proof, we assume two distinct time correlation coefficients $\rho_1$ and $\rho_2$ with $\rho_1 > \rho_2$. Denote the optimal solution to the original problem (\ref{eqn:joint_adaption}) under the time correlation of $\rho_1$ and the corresponding optimal energy efficiency as $({P_1}^*,\cdots,{P_L}^*,R^*)$ and $\eta_{I,L,\rho_1}^* = \eta_{I,L,\rho_1}({P_1}^*,\cdots,{P_L}^*,R^*)$, respectively. When the channel time correlation is reduced to $\rho_2$, the solution $({P_1}^*,\cdots,{P_L}^*,R^*)$ is still feasible to the problem (\ref{eqn:joint_adaption}) since it satisfies all the constraints in (\ref{eqn:joint_adaption}). With the decreasing monotonicity of $\eta_{I,L}$ with respect to $\rho$, we have $\eta_{I,L,\rho_1}({P_1}^*,\cdots,{P_L}^*,R^*)< \eta_{I,L,\rho_2}({P_1}^*,\cdots,{P_L}^*,R^*)$. Moreover, when the time correlation of $\rho_2$ is considered, the optimization in (\ref{eqn:joint_adaption}) will definitely lead to the optimal energy efficiency $\eta_{I,L,\rho_2}^*$ not lower than $\eta_{I,L,\rho_2}({P_1}^*,\cdots,{P_L}^*,R^*)$, i.e., $\eta_{I,L,\rho_2}^* \ge \eta_{I,L,\rho_2}({P_1}^*,\cdots,{P_L}^*,R^*)$.  Therefore, we have $\eta_{I,L,\rho_2}^* > \eta_{I,L,\rho_1}^*$. It means that the optimal energy efficiency ${{\eta}  _{I,L}^*}$ is a monotonically decreasing function of the time correlation $\rho$. Combining this monotonicity with (\ref{eqn:power_eff_opt_app}), we can conclude that $\theta_L$ is a decreasing functions of $\rho$ and the proof is finally completed.

\section{Proof of Property \ref{the:type_I_eff}}\label{app:there2_proof}
Consider two different maximal numbers of transmissions $L_1$ and $L_2$ with $L_1 \le L_2$. We denote the optimal transmission powers and rate and the corresponding optimal energy efficiency under the case with maximal $L_1$ transmissions as $\mathcal S_1^* = (P_1^*,\cdots,P_{L_1}^*,R^*)$ and $\eta_{L_1}^*$, respectively. Now when the maximal number of transmissions increases to $L_2$, we construct a solution of transmission powers and rate as $\mathcal S_2 = (P_1^*,\cdots,P_{L_1}^*, P_{L_1+1},\cdots, P_{L_2},R^*)$ where $ P_{L_1+1}=\cdots= P_{L_2}=0$. Clearly from the definition of outage probability in (\ref{outage_definition}), we have ${p_{out,{L_1}}} = \cdots =  {p_{out,{L_2}}}$. When the maximal number of transmissions is $L_2$, $\mathcal S_2$ constitues a feasible point of the problem (\ref{eqn:joint_adaption}), since $\mathcal S_2$ satisfies both outage and goodput constraints. Denote its corresponding energy efficiency as $\eta_{L_2}(\mathcal S_2)$. We directly have $\eta_{L_2}(\mathcal S_2)=\eta_{L_1}^*$. Through the optimization in (\ref{eqn:joint_adaption}), when the maximal number of transmissions is $L_2$, we can definitily find the optimal energy efficiency $\eta_{L_2}^*$ not lower than $\eta_{L_2}(\mathcal S_2)$, i.e., $\eta_{L_2}^* \ge \eta_{L_2}(\mathcal S_2)$. It then follows  $\eta_{L_2}^* \ge  \eta_{L_1}^*$, which means that the optimal energy efficiency is non-decreasing function of the maximum number of transmissions and thus completes the proof.
%

\section{Proof of Theorem \ref{the:ee_I}}\label{app:ee_I}
\subsection{Proof of (\ref{eqn:upper_bound})}
Taking the limit $L \to \infty$ of both sides of (\ref{eqn:power_eff_opt}) yields
\begin{align}\label{eqn:power_eff_upper_limi}
{\eta _{I,\infty}^*} &= \mathop {\lim }\limits_{L \to \infty } {{\psi {\theta _L}}}{{{{\mathcal T_0}}}}\frac{{{{\left( {{R^*} - {{\mathcal T_0}}} \right)}^c}}}{{{R^*}^c\left( {{2^{{R^*}}} - 1} \right)}}  \notag\\
&= \mathop {\lim }\limits_{L \to \infty } \frac{{\left( {2^{\frac{L}{{1 - {2^{ - L}}}} - 2}} \right){\theta _L}{{\mathcal T_0}}}}{{{{2^L} - 1}}}\mathop {\lim }\limits_{L \to \infty } \frac{{{{\left( {{R^*} - {{\mathcal T_0}}} \right)}^c}}}{{{R^*}^c\left( {{2^{{R^*}}} - 1} \right)}}\notag\\
 &= \frac{{{\theta _\infty }{{\cal T}_0}}}{4}\mathop {\lim }\limits_{L \to \infty } \frac{{{{\left( {{R^*} - {{\cal T}_0}} \right)}^c}}}{{{R^*}^c\left( {{2^{{R^*}}} - 1} \right)}},
\end{align}
where $\theta_\infty \triangleq {\mathop {\lim }\limits_{L \to \infty } {\theta _L}}$ and the existence of $\theta_\infty$ will be proved in Appendix \ref{app:proof_bound_ratio_ell}. As $L \to \infty$, we surely have $\Delta =\min\{\varepsilon,2^{-L}\}= 2^{-L}$. Recalling ${{{\mathcal T_0}} < R^* \le\frac{{{{\mathcal T_0}}}}{{1 - \Delta }}}$, $\mathop {\lim }\limits_{L \to \infty } {R^*} = {{\mathcal T_0}}$ follows by applying squeeze theorem. Plugging this into (\ref{eqn:power_eff_upper_limi}) along with $\mathop {\lim }\limits_{L \to \infty } {c} = {0}$  leads to
\begin{align}\label{eqn:power_eff_upper_limi1}
{\eta _{I,\infty}^*} & = \frac{\theta _\infty{{{\mathcal T_0}}}}{{4\left( {{2^{{{\mathcal T_0}}}} - 1} \right)}}{{\mathop {\lim }\limits_{L \to \infty } {{\left( {{R^*} - {{\mathcal T_0}}} \right)}^c}}}.
\end{align}
Noticing that both ${{R^*} - {{\mathcal T_0}}}$ and $c$ in (\ref{eqn:power_eff_upper_limi1}) converge to zero, the occurrence of the limit $0^0$ complicates the derivations. Fortunately, based on Theorem \ref{the:optimal_rate}, we have $R^* = \varphi^{-1}(0) $ when $L \to \infty$ as $\varepsilon$ becomes greater than $2^{-L}$.
It implies that $R^*$ is a zero point of $\varphi \left( {{R}} \right)$, i.e., $\varphi \left( {{R^*}} \right) = 0$. Following the definition of $\varphi \left( {{R}} \right) $, we then have
\begin{equation}\label{eqn:L_inf_R_stc4}
{R^*} = {{\mathcal T_0}} + c{{\mathcal T_0}}\Xi(R^*),
\end{equation}
where $\Xi(R)=\frac{{{2^R} - 1}}{{\ln \left( 2 \right)R{2^R}}}$. Now putting (\ref{eqn:L_inf_R_stc4}) into (\ref{eqn:power_eff_upper_limi1}) yields
\begin{align}\label{eqn:type_power_eff_fin_bound}
{\eta _{I,\infty}^*} &= \frac{{{{{\theta _\infty }}}{{\cal T}_0}}}{{4\left( {{2^{{{\cal T}_0}}} - 1} \right)}}{{\mathop {\lim }\limits_{L \to \infty } {{\left( {c{{\cal T}_0}\Xi ({R^*})} \right)}^c}}}\notag\\
&= \frac{{{{{\theta _\infty }}}{{\cal T}_0}}}{{4\left( {{2^{{{\cal T}_0}}} - 1} \right)}}{{\mathop {\lim }\limits_{L \to \infty } {c^c}{{\left( {\mathop {\lim }\limits_{L \to \infty } {{\cal T}_0}\Xi ({R^*})} \right)}^{\mathop {\lim }\limits_{L \to \infty } c}}}}\notag\\
 &= \frac{{{{{\theta _\infty }}}{{\cal T}_0}}}{{4\left( {{2^{{{\cal T}_0}}} - 1} \right)}}.
\end{align}
The result in (\ref{eqn:upper_bound}) then directly follows by combining (\ref{eqn:type_power_eff_fin_bound}) with Property \ref{the:type_I_eff}.

Clearly, ${\eta _{I,\infty}^*}$ is a decreasing function of the goodput threshold ${{\cal T}_0}$. Moreover, when ${\sigma_1}^2=\cdots={\sigma_L}^2=1$ and $\rho=0$, we have $\theta_L=1$ based on its definition in (\ref{eqn:theta_I}). As shown in Appendix \ref{app:time_corr}, $\theta_L$ is a decreasing function of the time correlation coefficient $\rho$. Therefore, when ${\sigma_1}^2=\cdots={\sigma_L}^2=1$ with $L \to \infty$, ${\theta _\infty }$ would be generally lower than or equal to that for the case with $\rho=0$. That is, ${\theta _\infty } \le 1$. Applying this inequality into the maximal energy efficiency (\ref{eqn:type_power_eff_fin_bound}) and using the decreasing monotonicity of ${\eta _{I,\infty}^*}$ with respect to ${{\cal T}_0}$, it is easy to get (\ref{type_I_maximal}).

\subsection{Existence of $\theta_\infty$}\label{app:proof_bound_ratio_ell}
Based on the definition in (\ref{eqn:theta_I}), $\theta _\infty$ can be written as
\begin{equation}\label{eqn:theta_infy_limit}
{\theta _\infty } = \mathop {\lim }\limits_{L \to \infty } \underbrace{\prod\limits_{k = 1}^L {{{\left( {{\sigma _k}^2} \right)}^{{2^{ - k}}}}}}_{\omega_L}  \mathop {\lim }\limits_{L \to \infty } \underbrace {\prod\limits_{k = 1}^L {{{\left( {\frac{{\ell \left( {k,\rho } \right)}}{{\ell \left( {k - 1,\rho } \right)}}} \right)}^{{2^{ - k}}}}} }_{{\vartheta _L}}.
\end{equation}
Here $\omega_L$ can be rewritten as $\omega_L  = \exp \left( {\sum\nolimits_{k = 1}^L {{2^{ - k}}\ln {\sigma _k}^2} } \right)$. Since the channel power gain ${\sigma_1}^2, \cdots, {\sigma_\infty}^2$ are generally bounded in practice and every absolutely convergent series is convergent \cite{rudin1964principles}, the exponent ${\sum\nolimits_{k = 1}^L {{2^{ - k}}\ln {\sigma _k}^2} }$ absolutely converges as $L \to \infty$. It follows that $\omega_L$ converges as $L \to \infty$, i.e., the limit of $\omega_L$ for $L \to \infty$ exists.
%

With respect to the term $\vartheta _L$, following the definition of $\ell \left( {k,\rho }\right) $ in (\ref{eqn:time_corr_impa}), we have
\begin{align}\label{eqn:ratio_ell_simpl}
\frac{{\ell \left( {k,\rho } \right)}}{{\ell \left( {k - 1,\rho } \right)}} 
&=\frac{{1 + \left( {1 - {\rho ^{2(k   - 1)}}} \right)\sum\limits_{t = 1}^{k - 1} {\frac{{{\rho ^{2(t   - 1)}}}}{{1 - {\rho ^{2(t   - 1)}}}}} }}{{1 + \sum\limits_{t = 1}^{k - 1} {\frac{{{\rho ^{2(t   - 1)}}}}{{1 - {\rho ^{2(t   - 1)}}}}} }} \le 1.
\end{align}
Accordingly, the sequence $\{\vartheta _L\}$ decreases with $L$ and is lower bounded as $\vartheta _L \ge 0$. Therefore, the limit $\mathop {\lim }\limits_{L \to \infty } {\vartheta _L}$ exists. The existence of both the limits $\mathop {\lim }\limits_{L \to \infty } {\vartheta _L}$ and $\mathop {\lim }\limits_{L \to \infty } {\omega_L}$ then guarantees the existence $\theta_\infty$ and completes the proof.

\section{Derivation of $\kappa_\infty$}\label{app:proof_rec}
Noticing ${{2^{ - k}}} > 0$ together with ${\kappa _L} = \prod\nolimits_{k = 1}^L {{{ {{k}} }^{{2^{ - k}}}}}$, $\kappa _L$ is readily found to be an increasing sequence, such that ${\kappa _1} < {\kappa _2} <  \cdots  < {\kappa _\infty }$.
To prove the convergence of the sequence ${\kappa _L }$, it suffices to prove that ${\kappa _\infty}$ is upper bounded. To this end, ${\kappa _\infty }$ is rewritten as ${\kappa _\infty } = {e^{ \sum\nolimits_{k = 1}^\infty  {{2^{ - k}}\ln \left( k \right)} }}$. Then applying Jensen's inequality yields
\begin{equation}\label{eqn:lower_rec_K_prod_fin}
{\kappa _\infty} \le {e^{ \ln \left( {\sum\limits_{k = 1}^\infty {k{2^{ - k}}} } \right)}} = {{\sum\limits_{k = 1}^\infty  {k{2^{ - k}}} }}.
\end{equation}
By using \cite[Eq.0.231]{gradshteyn1965table}, we have ${\kappa _\infty} \le {2}$. Accordingly, the sequence ${\kappa _{L}}$ is convergent, and the limit ${\kappa _\infty }= \mathop {\lim }\limits_{{L} \to \infty} \kappa_L$  exists. It is difficult to evaluate ${\kappa _\infty}$ exactly, however ${\kappa _\infty}$ can be approximated as ${\kappa _L}$ with a large $L$. To guarantee the approximate accuracy, $L$ should be properly chosen to achieve a sufficiently low approximation error which is defined as
\begin{equation}\label{eqn:comput_err_kappa}
{e_L} \triangleq \frac{{{\kappa _\infty } - {\kappa _L} }}{{{\kappa _\infty }}} = 1 - \prod\limits_{k = L + 1}^\infty  {{k^{{2^{ - k}}}}}.
\end{equation}
Clearly, the approximation error ${e_L}$ decreases with $L$, and is upper bounded as
\begin{align}\label{eqn:e_L_bound}
{e_L} &= 1 - {e^{{-2^{ - L}}\sum\limits_{k = L + 1}^\infty  {\frac{{{2^{ - k}}}}{{{2^{ - L}}}}\ln \left( k \right)} }}  \notag \\
&\le 1 - {e^{{-2^{ - L}}\ln \left( {\sum\limits_{k = L + 1}^\infty  {\frac{{{2^{ - k}}}}{{{2^{ - L}}}}k} } \right)}} \notag\\
&= 1 - {e^{{-2^{ - L}}\ln \left( {{\sum\limits_{k = 1}^\infty  {{2^{ - k}}k}  + L\sum\limits_{k = 1}^\infty  {{2^{ - k}}} } } \right)}} \notag\\
&= 1 - {\left( { {L + 2}} \right)^{{-2^{ - L}}}},
\end{align}
wherein the inequality holds by using Jensen's inequality. Based on (\ref{eqn:e_L_bound}), by setting $L=20$, the approximation error is upper bounded as ${e_L} \le 2.95*10^{-6}$. Thus ${\kappa _\infty}$ can be approximated as $\kappa _\infty \approx \kappa _{20}=1.6617$ with approximation error less than $2.95*10^{-6}$ which is sufficiently low for practical applications.

\section{Proof of Increasing Monotonicity of ${\cal G}_k\left( R \right)$}\label{app:proof_increasing_g_ratio}
To prove the increasing monotonicity of ${\cal G}_k\left( R \right)$, we resort to analyze its first order derivative of ${\cal G}_k\left( R \right)$ with respect to $R$, given by
\begin{align}\label{eqn:first_derive_cal_G}
{\cal G}_k'\left( R \right)
= \frac{{{\cal U}_k\left( R \right)}}{{{{\left( {R{g_{k - 1}}\left( R \right)} \right)}^2}}},
\end{align}
where $\mathcal U_k\left( R \right) = R{g_k}^\prime \left( R \right){g_{k - 1}}\left( R \right) - R{g_{k - 1}}^\prime \left( R \right){g_k}\left( R \right) - {g_{k - 1}}\left( R \right){g_k}\left( R \right)$ and ${g_k}^\prime(R)$ can be further written by using Residue theorem as
\begin{align}\label{eqn:g_K_fir_der}
{g_k}^\prime \left( R \right) &= \frac{{\ln \left( 2 \right)}}{{2\pi {\rm{i}}}}\int\nolimits_{{a} - {\rm{i}}\infty }^{{a} + {\rm{i}}\infty } {\frac{{{2^{Rs}}}}{{{{\left( {s - 1} \right)}^k}}}ds} \notag\\
&= \frac{{\ln \left( 2 \right){{\left( {R\ln \left( 2 \right)} \right)}^{k - 1}}}}{{\left( {k - 1} \right)!}}{2^R},\, k \ge 1.
\end{align}
Specifically when $k=1$, it follows from (\ref{eqn:g_l_def}) that $\mathcal U_k\left( R \right)$ reduces to $\mathcal U_1\left( R \right) = R\ln(2)2^R - 2^R+1 = {2^R}\left( {R\ln (2) - 1 + {e^{ - \ln \left( 2 \right)R}}} \right)$. Clearly by using the inequality $e^{-x} \ge 1-x$, we have $\mathcal U_1\left( R \right) \ge 0$. It means that ${\cal G}_k\left( R \right)$ is an increasing function when $k=1$. On the other hand, when $k \ge 2$, it becomes very difficult to determine the sign of $\mathcal U_k\left( R \right) $. To proceed, we turn to analyze the first order derivative of $\mathcal U_k\left( R \right)$ with respect to $R$ given as
\begin{multline}\label{eqn:first_der_U}
\mathcal U_k'\left( R \right) = R{g_k}^{\prime \prime }\left( R \right){g_{k - 1}}\left( R \right) \\
 - \left( {2{g_{k - 1}}^\prime \left( R \right)
+ R{g_{k - 1}}^{\prime \prime }\left( R \right)} \right){g_k}\left( R \right),
\end{multline}
where similarly to (\ref{eqn:g_K_fir_der}), the second order derivative of
${g_k}(R)$ can also be written by using Residue theorem as
\begin{multline}\label{eqn:g_K_fun_sec_der}
{g_k}^{\prime \prime }\left( R \right) = \frac{\left( {\ln 2} \right)^2}{2 \pi \rm i}\int\nolimits_{{a} - {\rm{i}}\infty }^{{a} + {\rm{i}}\infty } {\frac{{s{2^{Rs}}}}{{{{\left( {s - 1} \right)}^k}}}ds}= \frac{{{{\left( {\ln 2} \right)}^2}}}{{\left( {k - 1} \right)!}} \times\\
  \left( {{2^R}{{\left( {R\ln 2} \right)}^{k - 1}} + \left( {k - 1} \right){2^R}{{\left( {R\ln 2} \right)}^{k - 2}}} \right)
  ,k \ge 2.
\end{multline}
Substituting (\ref{eqn:g_K_fir_der}) and (\ref{eqn:g_K_fun_sec_der}) into (\ref{eqn:first_der_U}), it follows that
\begin{align}\label{eqn:first_deriv_U_fin}
\mathcal U_k^{\prime}\left( R \right) &= \frac{{\ln \left( 2 \right){{\left( {R\ln 2} \right)}^{k - 2}}{2^R}}}{{\left( {k - 2} \right)!}} \notag\\
&\times\underbrace {\left( \begin{array}{l}
\frac{{R\ln \left( 2 \right)}}{{\left( {k - 1} \right)}}\left( {R\ln \left( 2 \right) + \left( {k - 1} \right)} \right){g_{k - 1}}\left( R \right)\\
 - \left( {R\ln \left( 2 \right) + k} \right){g_k}\left( R \right)
\end{array} \right)}_{{\cal J}\left( R \right)}.
\end{align}
It can be proved that $\mathcal J(R)$ is an increasing and convex function, since its first and second order derivatives of $\mathcal J(R)$ with respect to $R$ satisfy
\begin{align}\label{eqn:J_R_deriv_first}
\mathcal J^{\prime}\left( R \right) 
& = \left( {\frac{{2R{{\left( {\ln \left( 2 \right)} \right)}^2}}}{{\left( {k - 1} \right)}} + \ln \left( 2 \right)} \right){g_{k - 1}}\left( R \right)  \notag\\
&\quad - \ln \left( 2 \right){g_k}\left( R \right) - \frac{{\ln \left( 2 \right){{\left( {R\ln \left( 2 \right)} \right)}^{k - 1}}}}{{\left( {k - 1} \right)!}}{2^R},
\end{align}
\begin{align}\label{eqn:J_R_secon_der}
\mathcal J^{\prime \prime}\left( R \right) 
&= \frac{{2{{\left( {\ln \left( 2 \right)} \right)}^2}}}{{\left( {k - 1} \right)}}{g_{k - 1}}\left( R \right) \ge 0.
\end{align}
Noticing that $\mathcal J^{\prime \prime}\left( R \right) \ge 0$ and $g_k(0)=0$, thus $\mathcal J^{\prime}(R) \ge \mathcal J^{\prime}(0) = 0$ for $k \ge 2$. Analogously, we can sequentially prove $\mathcal J(R) \ge \mathcal J(0) = 0$, $\mathcal U_k^{\prime}(R) \ge \mathcal U_k^{\prime}(0) = 0$, $\mathcal U_k(R) \ge \mathcal U_k(0) = 0$, then we have $\mathcal G_k^{\prime}(R) \ge 0$. Thus the increasing monotonicity of $\mathcal G_k(R)$ follows as well for $k\ge2$. The proof is eventually completed.

\section{Proof of Property \ref{the:bounds_ee_ir}} \label{app:proof_ratio_g_ineqs}
Prior to proving (\ref{eqn:Lambda_lowerbound}), the following upper and lower bounds associated with $\frac{{{g_k}\left( R \right)}}{{{g_{k - 1}}\left( R \right)}}$ are obtained first.
\subsection{$\frac{{{g_k}\left( R \right)}}{{{g_{k - 1}}\left( R \right)}} \le \frac{{R\ln \left( 2 \right)}}{{k - 1}}$}
To prove this inequality, it suffices to show that ${\frak{A}}(R) \triangleq \frac{{R\ln \left( 2 \right)}}{{k - 1}}{g_{k - 1}}\left( R \right) - {g_k}\left( R \right) \ge 0$.
By taking the first order derivative of ${\frak{A}}(R)$ with respect to $R$, it follows that
\begin{align}\label{eqn:g_upper_A_1st_der}
\frak A'\left( R \right) &= \frac{{\ln \left( 2 \right)}}{{k - 1}}{g_{k - 1}}\left( R \right) + \frac{{R\ln \left( 2 \right)}}{{k - 1}}{g_{k - 1}}^\prime \left( R \right) - {g_k}^\prime \left( R \right)\notag\\
& = \frac{{\ln \left( 2 \right)}}{{k - 1}}{g_{k - 1}}\left( R \right) \ge 0.
\end{align}
where the second equality holds by using (\ref{eqn:g_K_fir_der}). Accordingly, ${\frak{A}}(R)$ is an increasing function, which implies ${\frak{A}}(R) \ge {\frak{A}}(0) = 0$. Then the upper bound of $\frac{{{g_k}\left( R \right)}}{{{g_{k - 1}}\left( R \right)}}$ holds.
\subsection{$\frac{{{g_k}\left( R \right)}}{{{g_{k - 1}}\left( R \right)}} \ge \frac{{R\ln \left( 2 \right)}}{k}$} \label{lower_bound}
To prove this lower bound, it is equivalent to prove $\frak{B}(R) \triangleq \frac{{R\ln \left( 2 \right)}}{k}{g_{k - 1}}\left( R \right) - {g_k}\left( R \right) \le 0$.
To this end, taking the first order derivative of $\frak{B}(R)$ with respect to $R$, we have
\begin{align}\label{eqn:g_lower_bound_B_frist_de}
\frak B'\left( R \right) 
 &= \frac{{\ln \left( 2 \right)}}{k}{g_{k - 1}}\left( R \right) - \frac{{\ln \left( 2 \right){{\left( {R\ln \left( 2 \right)} \right)}^{k - 1}}{2^R}}}{{k!}}.
\end{align}
Then taking the second order derivative of $\frak{B}(R)$ with respect to $R$ yields
\begin{align}\label{eqn:B_second_der}
\frak B^{\prime \prime}\left( R \right) 
 &=  - \frac{{{{\left( {\ln \left( 2 \right)} \right)}^2}{{\left( {R\ln \left( 2 \right)} \right)}^{k - 1}}{2^R}}}{{k!}} \le 0.
\end{align}
Combining (\ref{eqn:g_lower_bound_B_frist_de}) and (\ref{eqn:B_second_der}), it is clearly found that $\frak B\left( R \right)$ is a decreasing function, which indicates $\frak{B}(R) \le \frak{B}(0) = 0$. Thus the lower bound of $\frac{{{g_k}\left( R \right)}}{{{g_{k - 1}}\left( R \right)}}$ is proved.

As a consequence, substituting the upper and lower bounds of $\frac{{{g_k}\left( R \right)}}{{{g_{k - 1}}\left( R \right)}}$ into $\Lambda \left( R \right)$ directly yields the upper and lower bounds of $\Lambda \left( R \right)$, respectively, as shown in (\ref{eqn:Lambda_lowerbound}).


\section{Proof of (\ref{eqn:equa_ratio_upper_comp_cc_ir})}\label{app:g_k_thri_ine}
For the proof, it suffices to show that $\frak C (R) \triangleq {g_k}\left( R \right) - \frac{{{2^R} - 1}}{k}{g_{k - 1}}\left( R \right) \le 0$.
Taking the first order derivative of $\frak C (R)$ with respect to $R$ gives
\begin{multline}\label{eqn:ratio_g_comp_first_der}
\frak C^ {\prime} \left( R \right)=
{2^R}\ln \left( 2 \right)\left( {\frac{{{{\left( {R\ln \left( 2 \right)} \right)}^{k - 1}}}}{{\left( {k - 1} \right)!}} - } \right.\\
\left. {\frac{{k - 1}}{k}\frac{{\left( {{2^R} - 1} \right){{\left( {R\ln \left( 2 \right)} \right)}^{k - 2}}}}{{\left( {k - 1} \right)!}} - \frac{{{g_{k - 1}}\left( R \right)}}{k}} \right).
\end{multline}
By using the lower bound of $\frac{{{g_k}\left( R \right)}}{{{g_{k - 1}}\left( R \right)}}$ in Appendix \ref{lower_bound}, namely $\frac{{{g_k}\left( R \right)}}{{{g_{k - 1}}\left( R \right)}} \ge \frac{{R\ln \left( 2 \right)}}{k}$, it follows that
\begin{multline}\label{eqn:lower_bound_g_k_1}
{g_{k - 1}}\left( R \right) \ge \frac{{R\ln \left( 2 \right)}}{{k - 1}}{g_{k - 2}}\left( R \right) \ge
 \cdots  \\
 \ge \frac{{{{\left( {R\ln \left( 2 \right)} \right)}^{k - 2}}}}{{\left( {k - 1} \right)!}}{g_1}\left( R \right) = \frac{{{{\left( {R\ln \left( 2 \right)} \right)}^{k - 2}}\left( {{2^R} - 1} \right)}}{{\left( {k - 1} \right)!}}.
\end{multline}
Substituting (\ref{eqn:lower_bound_g_k_1}) into (\ref{eqn:ratio_g_comp_first_der}) yields
\begin{align}\label{eqn:ineq_C_frist_bound}
\frak C'\left( R \right) 
&\le \frac{{{2^R}\ln \left( 2 \right){{\left( {R\ln \left( 2 \right)} \right)}^{k - 2}}}}{{\left( {k - 1} \right)!}}\left( {R\ln \left( 2 \right) - \left( {{2^R} - 1} \right)} \right) \le 0.
\end{align}
Thus $\frak C (R)$ turns out to be a decreasing function of $R$, which implies $\frak C (R) \le \frak C (0) = 0$. The proof is then completed.

\bibliographystyle{ieeetran}
\bibliography{manuscript_1}

\end{document}